%% file: main.tex
\documentclass[sigconf]{acmart}

\usepackage{balance}

\def\BibTeX{{\rm B\kern-.05em{\sc i\kern-.025em b}\kern-.08emT\kern-.1667em\lower.7ex\hbox{E}\kern-.125emX}}

\usepackage{balance}
\usepackage{flushend}
\usepackage{amssymb}
\usepackage{setspace}
\usepackage{amsfonts,amsmath,amsthm}
\usepackage{latexsym,fancyhdr,url}
\usepackage{qtree}
\usepackage{mdframed}
\usepackage{mathtools}
\usepackage{algorithm2e}
\usepackage{algpseudocode}
\usepackage{pdfpages}
 \usepackage{graphics}
\usepackage{multirow}
\usepackage{tabularx, makecell}
\usepackage{enumitem}

\usepackage[n,advantage,operators,sets,adversary,landau,probability,notions,logic,
 ff,mm,primitives,events,complexity,asymptotics,keys]{cryptocode}

\usepackage{
  tikz,
  pgfplots,
  pgfplotstable,
}
\RequirePackage{fix-cm}

\usepackage{booktabs}
\usetikzlibrary{
  shapes.geometric,
  arrows,
  external,
  pgfplots.groupplots
}

\pgfplotsset{
  compat=1.9,
  every tick label/.append style={font=\scriptsize}
}

\DeclareMathAlphabet{\mathcal}{OMS}{cmsy}{m}{n}
\input{tomsdefs}
\def\authnote{1}
\renewcommand\footnotetextcopyrightpermission[1]{}
\begin{document}

\fancyhead{}

\title{Protocols for Checking Compromised Credentials}

\author{Lucy Li}
\affiliation{
\institution{Cornell University}
}

\author{Bijeeta Pal}
\affiliation{\institution{Cornell University}}

\author{Junade Ali}
\affiliation{\institution{Cloudflare Inc.}}

\author{Nick Sullivan}
\affiliation{\institution{Cloudflare Inc.}}

\author{Rahul Chatterjee}
\affiliation{
	\institution{University of Wisconsin--Madison \& Cornell Tech}
}

\author{Thomas Ristenpart}
\affiliation{
\institution{Cornell Tech}
}

\renewcommand{\shortauthors}{Li et al.}
\input{abstract}
\maketitle

\input{intro}
\input{overview}

\input{preliminaries}

\input{HIBP}
\input{FSBP}

\input{simulation}

\input{performance}
\input{discussions}
\input{relatedwork}
\input{conclusion}
\input{ack}

\bibliographystyle{ACM-Reference-Format}   
\bibliography{bib}
\appendix
\input{appendix}

\input{fsbp-proof}
\input{correlated}

\end{document}

%% file: tomsdefs.tex
\newcommand{\bnm}{\begin{newmath}}
\newcommand{\enm}{\end{newmath}}

\newcommand{\bea}{\begin{neweqnarrays}}%
\newcommand{\eea}{\end{neweqnarrays}}%

\newcommand{\bne}{\begin{newequation}}
\newcommand{\ene}{\end{newequation}}

\newcommand{\bal}{\begin{newalign}}
\newcommand{\eal}{\end{newalign}}

\newenvironment{newalign}{\begin{align*}%
\setlength{\abovedisplayskip}{4pt}%
\setlength{\belowdisplayskip}{4pt}%
\setlength{\abovedisplayshortskip}{6pt}%
\setlength{\belowdisplayshortskip}{6pt} }{\end{align*}}

\newenvironment{newmath}{\begin{displaymath}%
\setlength{\abovedisplayskip}{4pt}%
\setlength{\belowdisplayskip}{4pt}%
\setlength{\abovedisplayshortskip}{6pt}%
\setlength{\belowdisplayshortskip}{6pt} }{\end{displaymath}}

\newenvironment{neweqnarrays}{\begin{eqnarray*}%
\setlength{\abovedisplayskip}{4pt}%
\setlength{\belowdisplayskip}{4pt}%
\setlength{\abovedisplayshortskip}{4pt}%
\setlength{\belowdisplayshortskip}{4pt}%
\setlength{\jot}{0.0in} }{\end{eqnarray*}}

\newenvironment{newequation}{\begin{equation}%
\setlength{\abovedisplayskip}{4pt}%
\setlength{\belowdisplayskip}{4pt}%
\setlength{\abovedisplayshortskip}{6pt}%
\setlength{\belowdisplayshortskip}{6pt} }{\end{equation}}

\newcounter{ctr}

\newenvironment{newenum}{%
\begin{list}{{\rm (\arabic{ctr})}\hfill}{\usecounter{ctr} \labelwidth=17pt%
\labelsep=3pt \leftmargin=21pt \topsep=3pt%
\setlength{\listparindent}{\saveparindent}%
\setlength{\parsep}{\saveparskip}%
\setlength{\itemsep}{2pt} }}{\end{list}}

\newlength{\saveparindent}
\newlength{\saveparskip}
\newcommand{\E}{{\rm I\kern-.3em E}}
\newcommand{\Prob}[1]{{\Pr\left[#1\right]}}
\newcommand{\CondProb}[2]{\Pr\left[#1\ \left|\right.\ #2\right]}

\newcommand{\given}{\ensuremath{\,\big|\,}}

\newcommand{\secref}[1]{\mbox{Section~\ref{#1}}}
\newcommand{\appref}[1]{\mbox{Appendix~\ref{#1}}}
\newcommand{\thref}[1]{\mbox{Theorem~\ref{#1}}}

\newcommand{\figref}[1]{\mbox{Figure~\ref{#1}}}

\renewcommand{\eqref}[1]{\mbox{Equation~(\ref{#1})}}

\newcommand{\getsr}{{\:{\leftarrow{\hspace*{-3pt}\raisebox{.75pt}{$\scriptscriptstyle\$$}}}\:}}

\newcommand{\gamesfontsize}{\small}
\newcommand{\fpage}[2]{\framebox{\begin{minipage}[t]{#1\textwidth}\setstretch{1.1}\gamesfontsize  #2 \end{minipage}}}

\newcommand{\hfpages}[3]{\hfpagess{#1}{#1}{#2}{#3}}
\newcommand{\hfpagess}[4]{
        \begin{tabular}[t]{c@{\hspace*{.5em}}c}
        \framebox{\begin{minipage}[t]{#1\textwidth}\setstretch{1.1}\gamesfontsize #3 \end{minipage}}
        &
        \framebox{\begin{minipage}[t]{#2\textwidth}\setstretch{1.1}\gamesfontsize #4 \end{minipage}}
        \end{tabular}
    }

\newcommand{\Z}{\mathbb{Z}}

\renewcommand{\key}{{\small \kappa}}

\newcommand{\myind}{\hspace*{1em}}
\newcommand{\thh}{^{\textit{th}}} 
\renewcommand{\concat}{\,\|\,}

\def \part {part}

\newcommand{\pw}{w}

\newcommand{\advA}{{\mathcal{A}}}

\newcommand{\advB}{{\mathcal{A'}}} %

\DeclareMathOperator*{\argmax}{argmax}

\renewcommand{\paragraph}[1]{\vspace*{6pt}\noindent\textbf{#1}\;}

\newcounter{mytable}
\def\mytable{\begin{centering}\refstepcounter{mytable}}
\def\endmytable{\end{centering}}

\newcounter{myfig}
\def\myfig{\begin{centering}\refstepcounter{myfig}}
\def\endmyfig{\end{centering}}

\def \blackslug{\hbox{\hskip 1pt \vrule width 4pt height 8pt
    depth 1.5pt \hskip 1pt}}
\def \qed{\quad\blackslug\lower 8.5pt\null\par}

\def \zo  {\{0,1\}}

\newcommand\ignore[1]{}

\newcounter{mynote}[section]
\newcommand{\notecolor}{blue}
\newcommand{\thenote}{\thesection.\arabic{mynote}}
\newcommand{\tnote}[1]{\ifnum\authnote=1\refstepcounter{mynote}{\bf \textcolor{\notecolor}{$\ll$TomR~\thenote: {\sf #1}$\gg$}}\fi}

\newcommand{\fixme}[1]{\ifnum\authnote=1{\textcolor{red}{[FIXME: #1]}}\fi}
\newcommand{\better}[1]{\ifnum\authnote=1{\textcolor{violet}{[BetterWord: #1]}}\fi}
\newcommand{\todo}[1]{\ifnum\authnote=1{\textcolor{red}{[TODO: #1]}}\fi}
\newcommand{\point}[1]{\ifnum\authnote=1{\textcolor{gray}{/* #1  */}}\fi}

\newcounter{rcnote}[section]
\newcommand{\rcnote}[1]{\ifnum\authnote=1\refstepcounter{rcnote}{\bf \textcolor{magenta}{$\ll$RC~\thenote: {\sf #1}$\gg$}}\fi}
\newcommand{\mytab}{\hspace*{.4cm}}

\DeclareMathSymbol{\mlq}{\mathord}{operators}{``}
\DeclareMathSymbol{\mrq}{\mathord}{operators}{`'}

\renewcommand{\comment}[1]{\ifnum\authnote=1\textcolor{gray}{/* #1 */}\fi}

\newcommand{\genfrom}[1]{{\:{\leftarrow{\hspace*{-3pt}\raisebox{.75pt}{$\scriptscriptstyle#1$}}}\:}}

\newcommand{\rhf}[2]{R_{f, \gamma}}

\def\H{\ensuremath{\mathsf{H}}\xspace}
\newcommand{\Ht}[1]{\ensuremath{\H^{(#1)}\xspace}}
\newcommand{\Htl}{\Ht{\prefixlen}}

\renewcommand{\dist}{\mathcmd{p}}

\newcommand{\distw}{\dist_w}
\newcommand{\dists}{\ensuremath{\distest_s}}
\newcommand{\distn}{\ensuremath{\hat{\dist}_n}}

\newcommand{\distest}{\ensuremath{\hat{\mathcmd{p}}}}
\newcommand{\distbreach}{\ensuremath{\hat{\dist}_{h}}}
\newcommand{\w}{{w}}

\def\Bfsb{\B_{\textnormal{\ouralgo}}}
\def\Abarfsb{\Abar_{\ouralgo}}
\def\qdist{\dists(w_{\qbar})}

\newcommand{\secloss}{\Delta}

\newcommand{\mathcmd}[1]{\ensuremath{#1}\xspace}

\newcommand{\s}{\mathcmd{s}}
\newcommand{\typo}{\mathcmd{\tilde{\w}}}

\newcommand{\PW}{\mathcal{W}}
\newcommand{\users}{\mathcal{U}}
\newcommand{\id}{u}
\newcommand{\user}{u}

\newcommand{\creds}{\mathcal{S}}
\newcommand{\cred}{s}

\newcommand{\breachDB}{\ensuremath{\tilde{\creds}}}

\newcommand{\dblen}{N}

\newcommand{\testdata}{\ensuremath{T}}
\newcommand{\sptestdata}{\ensuremath{T_{\text{sp}}}}

\newcommand{\guesses}{\ensuremath{L}}

\newcommand{\Adv}{\textnormal{\textsf{Adv}}}
\newcommand{\AdvGame}[2]{\Adv^{\footnotesize\textnormal{\textrm{#1}}}_{{\footnotesize #2}}}

\newcommand{\zxcvbn}{\textnormal{\textsf{zxcvbn}}\xspace}

\newcommand{\hlen}{\ell}
\newcommand{\prefixlen}{l}

\newcommand{\bnote}[1]{\ifnum\authnote=1\refstepcounter{mynote}{\bf \textcolor{\notecolor}{$\ll$Bijeeta~\thenote: {\sf #1}$\gg$}}\fi}

\newcommand{\wss}[1]{{\ensuremath{\w_1,\ldots,\w_{#1}}}}
\newcommand{\uwss}[1]{{\ensuremath{(u_1,\w_1),\ldots,(u_{#1},\w_{#1})}}}

\newcommand{\q}{\ensuremath{q}}
\newcommand{\qbar}{\ensuremath{\bar{q}}}

\def\bucketset{\mathcal{B}}

\newcommand{\B}{\ensuremath{{\beta}}}
\def\nbuckets{|\bucketset|}

\newcommand{\A}{{\alpha}}
\newcommand{\Abar}{\tilde{\A}^{\vphantom{a}}}

\def\bw{\gamma}

\newcommand{\Guessgzero}{\textsf{Guess}}
\newcommand{\Guessgone}{\textsf{BucketGuess}}
\newcommand{\advone}[1]{\AdvGame{\textsf{b-gs}}{#1}}
\newcommand{\advzero}{\AdvGame{\textsf{gs}}{ }}
\newcommand{\zbf}{\ensuremath{\mathbf{z}}}

\newcommand{\hibp}{\text{HIBP}\xspace}

\newcommand{\fsbp}{\text{FSB}\xspace}
\newcommand{\ouralgo}{\fsbp}
\newcommand{\google}{\text{GPC}\xspace}
\newcommand{\gpc}{\google}
\newcommand{\idbp}{\text{IDB}\xspace}
\newcommand{\hpbp}{\text{HPB}\xspace}

\newcommand{\ccc}{\text{C3}\xspace}

\newcommand{\lnote}[1]{\ifnum\authnote=1\refstepcounter{mynote}{\bf \textcolor{orange}{$\ll$Lucy~\thenote: {\sf #1}$\gg$}}\fi}
\newcommand{\corrg}{\textsf{Corr-Guess}}

\newcommand{\transform}{\tau}

\newcommand{\powerset}[1]{\mathcal{P}\left( #1 \right)}


%% file: abstract.tex
\begin{abstract}
  To prevent credential stuffing attacks, industry best practice now proactively
  checks if user credentials are present in known data breaches.  Recently, some
  web services, such as HaveIBeenPwned (\hibp) and Google Password Checkup (\google),
  have started providing APIs to check for breached passwords. We refer to such
  services as \emph{compromised credential checking} (\ccc) services.
  We give the first formal description of \ccc services, detailing different settings and
  operational requirements, and we give relevant threat models.

  One key security requirement is the secrecy of a user's
  passwords that are being checked.
  Current widely deployed \ccc services have the user share a
  small prefix of a hash computed over the user's password. We provide a
  framework for empirically analyzing the leakage of such protocols, showing 
  that in some contexts knowing the hash prefixes leads to a 12x increase in the efficacy of
  remote guessing attacks.
  We propose two new protocols that provide stronger protection for users'
  passwords, implement them, and show experimentally that they remain practical
  to deploy.
\end{abstract}


%% file: intro.tex
\section{introduction}
\label{sec:intro}

Password database breaches have become routine~\cite{wiki-data-leaks}. Such
breaches enable credential stuffing attacks, in which attackers try to compromise 
accounts by submitting one or more passwords that were leaked with that account from another
website.  To counter credential stuffing,  companies and
other organizations have begun 
checking if their users' passwords appear in breaches, and, if so, they deploy
further protections (e.g., resetting the user's passwords or otherwise warning the
user). Information on what usernames and passwords have appeared in breaches is
gathered either from public sources or from a third-party service.  The latter
democratizes access to leaked credentials, making it easy for others to help
their customers gain confidence that they are not using exposed passwords.  We
refer to such services as \emph{compromised credential checking} services, or
\ccc services in short.

Two prominent \ccc services already operate. HaveIBeenPwned (\hibp)~\cite{HIBP}
was deployed by Troy Hunt and CloudFlare in 2018 and is used by many web services, including
Firefox~\cite{leakcheck:firefoxmonitor}, EVE Online~\cite{eveonline}, and
1Password~\cite{onepwHIBP}.  Google released a Chrome extension called
Password Checkup (\google)~\cite{gpc-blog:2019,thomas2019protecting} in  2019
that allows users to check if their username-password pairs appear in a
compromised dataset.  Both services work by having the user share with the \ccc
server a prefix of the
hash of their password or of the hash of their username-password pair. This leaks
some information about user passwords, which is problematic should the \ccc server be
compromised or otherwise malicious. 
But until now there has been no thorough investigation into the damage from the leakage of
current \ccc services or suggestions for protocols that provide better privacy. 

We provide the first formal treatment of \ccc services for different settings, 
including an exploration of their security guarantees.
A \ccc service must provide secrecy of client credentials, and
ideally, it should also preserve secrecy of the leaked datasets held by the \ccc
server.  The computational and bandwidth overhead for the client and especially
the server should also be low. 
The server might hold billions of leaked records, precluding use of
existing cryptographic protocols for private set intersection
(PSI)~\cite{freedman2004efficient,meadows1986more}, which would use a prohibitive amount of bandwidth 
at this scale. 

Current industry-deployed \ccc services  reduce bandwidth requirements by dividing the
leaked dataset 
into buckets before executing a PSI protocol.
The client shares with the \ccc server the identifier of the bucket where their
credentials would be found, if present in the leak dataset. Then, the client and the
server engage in a protocol between the bucket held by
the server and the credential held by the client to determine if their
credential is indeed in the leak.  In current schemes, the
prefix of the hash of the user credential is used as the bucket identifier. The
client shares the hash prefix (bucket identifier) of their credentials with the
\ccc server. 

Revealing hash prefixes of credentials may be dangerous. We outline
an attack scenario against such prefix-revealing \ccc services. In
particular, we consider a conservative setting where
the \ccc server attempts to guess the password, while knowing the username and the hash prefix associated with the queried credential. We 
rigorously evaluate the security of \hibp and \google under this threat model
via a mixture of formal and empirical analysis.

We start by considering users with a password appearing in some leak and show how to
adapt a recent state-of-the-art credential tweaking attack~\cite{pal2019beyond} 
to take advantage of the knowledge of hash prefixes. In a credential tweaking
attack, one uses the leaked password to determine likely guesses (usually, small
tweaks on the leaked password). Via simulation, we show that our variant of
credential tweaking successfully compromises~$83\%$ of such accounts with 1,000 or fewer attempts, given the
transcript of a query made to the \hibp server. Without knowledge of the transcript, only 56\% of these accounts can be compromised within 1,000 guesses.

We also consider user accounts not present in a leak. 
Here we found that the leakage from the hash prefix disproportionately affects
security compared to the previous case.
For these user accounts, 
obtaining the query to \hibp enables the attacker to guess 71\% 
of passwords within 1,000 attempts, which is a 12x increase over the success with no hash prefix 
information. Similarly, for \google, our
simulation shows $33\%$ of user passwords can be guessed in $10$ or fewer
attempts (and 60\% in 1,000 attempts), should the attacker learn the hash prefix shared with the \google
server. 

The attack scenarios described are conservative because they assume the attacker
can infer which queries to the \ccc server are associated to which usernames.
This may not be always possible. Nevertheless, caution dictates that we would
prefer schemes that leak
less. We therefore present two new \ccc protocols, one that checks for leaked passwords (like
\hibp) and one that checks for leaked username-password pairs (like \google).
Like \google and \hibp, we \emph{partition} the password space before performing
PSI, but we do so in a way that reduces leakage significantly. 

Our first scheme works when only passwords are queried to the \ccc server. 
It utilizes a novel approach that we call frequency-smoothing bucketization (\fsbp). The
key idea is to use an estimate of the distribution of human-chosen passwords to assign passwords to buckets in a way that flattens the distribution
of accessed buckets.  We show how to obtain good estimates (using leaked data),
and, via simulation,
that \fsbp reduces leakage significantly (compared to \hibp). In many cases the best
attack given the information leaked by the \ccc protocol works no better than having no information
at all.
While the benefits come with some added computational complexity and bandwidth,
we show via experimentation that the operational overhead for the \fsbp \ccc
server or client is comparable with the overhead from $\google$, while also leaking much less information than hash-prefix-based C3 protocols.

We also describe a more secure bucketizing scheme that provides better
privacy/bandwidth trade-off for \ccc servers that store username-password pairs.
This scheme was also (independently) proposed
in~\cite{thomas2019protecting}, and Google states that they plan to transition to using it in
their Chrome extension. 
It is a simple modification of their current protocol. We refer to it as
\idbp, ID-based bucketization, as it uses the hash prefix of only the user
identifier for bucketization (instead of the hash prefix of the
username-password pair, as currently used by \gpc). Not having password information in the
bucket identifier hides the user's password perfectly from an attacker who
obtains the client queries (assuming that passwords are independent of
usernames). We implement \idbp and show that the average bucket size in this
setting for a hash prefix of 16~bits is similar to that of \google (average
16,122 entries per bucket, which leads to a bandwidth of 1,066~KB).

\paragraph{Contributions.}
In summary, the main contributions of this paper are the following: 
\begin{itemize}[itemsep=3pt]
\item We provide a formalization of \ccc protocols and detail the security goals for such
  services.  
\item We discuss various threat models for \ccc services, and analyze the security
  of two widely deployed \ccc protocols. We show that an
  attacker that learns the queries from a client can severely damage the
  security of the client's passwords, should they also know the client's username.
\item We give a new \ccc protocol (\ouralgo) for checking only leaked passwords that
  utilizes knowledge of the human-chosen password distribution to reduce 
  leakage. 
\item We give a new \ccc protocol for checking leaked username-password pairs (\idbp) that bucketizes using only usernames. 
\item We analyze the performance and security of both new \ccc protocols to show
  feasibility in practice.
\end{itemize}
We will release as public, open source code our server and client
implementations of \fsbp and \idbp.



%% file: overview.tex
\section{Overview}
\label{sec:overview}

We investigate approaches to checking credentials present in previous breaches.  
Several third party services provide credential checking, enabling
 users and companies to mitigate credential stuffing and credential
tweaking
attacks~\cite{pal2019beyond,das2014tangled,wang2016targeted}
, an increasingly daunting problem for account security.

To date, such C3 services have not received in-depth security analyses. We start by describing the
architecture of such services, and then we detail relevant threat
models.

\paragraph{C3 settings.} We provide a diagrammatic summary of the abstract
architecture of C3 services in~\figref{fig:diagram}. A C3 \emph{server} has
access to a breach database $\breachDB$. We can think of $\breachDB$ as a set of
size~$\dblen$, which consists of either a set of passwords $\{\wss{\dblen}\}$  or username-password pairs $\{\uwss{\dblen}\}$.
This corresponds to two types of C3 services --- \emph{password-only C3 services}
and \emph{username-password C3 services}.  For example, 
\hibp~\cite{HIBP:v2} is a password-only C3 service,\footnote{$\hibp$
  also allows checking if a user identifier (email) is leaked with a data
  breach. We focus on password-only and username-password C3 services.} and Google's service
\gpc~\cite{gpc-blog:2019} is an example of a username-password C3 service. 

\begin{figure}[t]
  \center
  \includegraphics[width=0.4\textwidth]{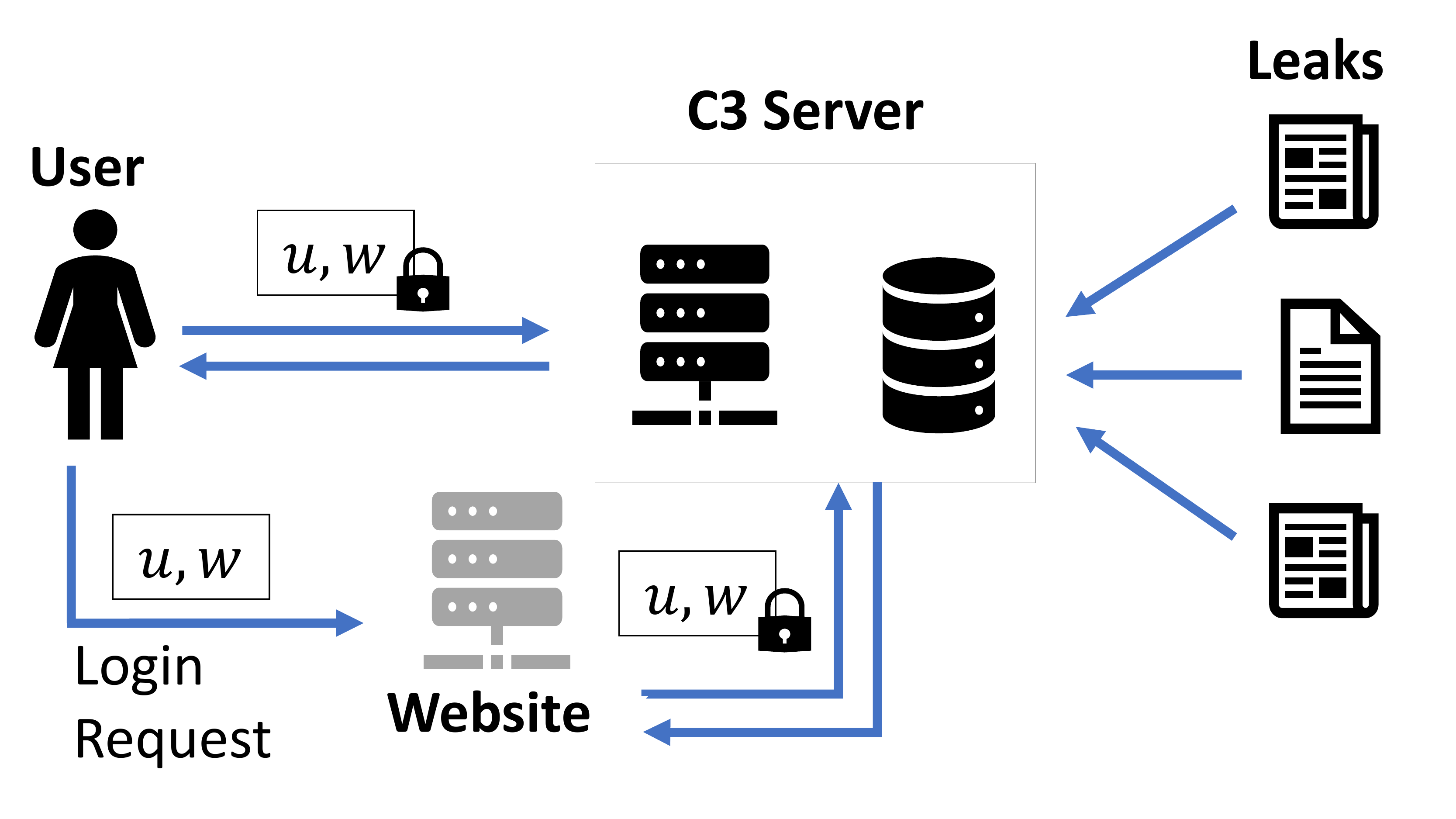}
  \caption{A \ccc service allows a client to ascertain whether a username and
    password appear in public breaches known to the service.}
  \label{fig:diagram}
\end{figure}

A \emph{client} has as input a credential $\cred=(\user,\w)$ and wants to determine if
$\cred$ is at risk due to exposure. The client and server therefore engage in a
set membership protocol to determine if $\cred \in \breachDB$. Here,
clients can be users themselves (querying the C3 service using, say, a browser extension), or
other web services can query the C3 service on behalf of their users. Clients
may make multiple queries to the C3 service, though the number of queries might
be rate limited.

The ubiquity of breaches means that, nowadays, the breach database $\breachDB$
will be quite large. A recently leaked compilation of previous breached data
contains $1.4$~billion username password pairs~\cite{leakurl}.  The
HIBP database has $501$ million unique passwords
~\cite{HIBP:v2}. Google's blog specifies that there are 4~billion
username-password pairs in their database of leaked credentials~\cite{gpc-blog:2019}. 

\ccc protocols should be able to scale to handle set membership requests for
these huge datasets for millions of requests a day. \hibp reported serving around 600,000 requests per day on average \cite{cloudflare:blog}.
The design of a C3 protocol should therefore not be
expensive for the server. Some clients may have limited computational power, so the C3 protocol should also not be expensive on the client-side. 
The number of network round trips
required must be low, and we  restrict attention
to protocols that can be completed with a single HTTPS request. Finally, we will
want to minimize bandwidth usage.

\paragraph{Threat model.} We consider the security of \ccc protocols relative to
two distinct threat models: (1) a malicious client that
wants to learn a different user's password; and (2) an honest-but-curious \ccc server that aims to learn
the password corresponding to a \ccc query. We discuss each in turn. 

A malicious client may want to use the \ccc server to discover another user's
password.  The malicious client may know the target's username and has the
ability to query the \ccc server.  The \ccc server's database $\breachDB$ should
therefore be considered confidential, and our security goal here is that each
query to the \ccc server can at most reveal whether a particular $\w$ or
$(\user,\w)$ is found within the breach database, for password-only and
username-password services, respectively. Without some way of authenticating
ownership of usernames, this seems the best possible way to limit knowledge gained from queries. We note that most breach
data is in fact publicly available, so we should assume that dedicated adversaries in this threat
model can find (a substantial fraction of) any \ccc
service's dataset. For such adversaries, there is little value in attempting to
exploit the \ccc service via queries. Nevertheless, deployments should
rate-limit clients via IP-address-based query throttling as well as via
slow-to-compute hash functions such as Argon2~\cite{argon2}.

The trickier threat model to handle is (2), and this will consume most of our
attention in this work. Here the \ccc server may be compromised or otherwise malicious,
and it attempts to exploit a client's queries to help it learn that client's
password for some other target website. We assume the adversary can submit 
password guesses to the target website, and that it knows the client's username. 
We refer to this setting as a known-username attack (KUA).
We conservatively\footnote{This is conservative because the \ccc server need
not, and should not, store passwords in-the-clear, and it should instead obfuscate 
them using an oblivious PRF.}
assume the adversary has access to the full breach dataset, and thus can
take advantage of both leaked passwords available in the breach dataset 
and  information leaked about the client's password from \ccc queries.
Looking ahead, for our protocols, the information potentially leaked from \ccc queries is the bucket identifier.

It is context-dependent whether a compromised \ccc server will be able to
mount KUAs. For example, in deployments where a web server
issues queries on behalf of their users, queries associated to
many usernames may be intermingled. In some cases, however, an adversary may be
able to link usernames to queries by observing meta-data corresponding to a query
(e.g., IP address of the querying user or the timing of a request). One can imagine cross-site scripting attacks that
somehow trigger requests to the C3 service, or the adversary might send 
tracking emails to leaked email addresses in order to 
infer an IP address associated to a username~\cite{englehardt2018never}. 
We therefore conservatively assume the malicious server's ability to know the
correct username for a query.  

In our KUA model, we focus on online attack settings, where the attacker tries
to impersonate the target user by making remote login attempts at another web
service, using guessed passwords. 
These are easy to launch and are one of the most prevalent forms of
attacks~\cite{4iqreport,enterprise20172017}.  However, in an online setting, the
web service should monitor failed login attempts and lock an account after too
many incorrect password submissions. Therefore, the attacker gets only a small
number of attempts. We use a variable $\q$, called the guessing budget, to
represent the allowed number of attempts.

Should the adversary additionally have access to password hashes stolen from the
target web site, they can instead mount an offline cracking attack. Offline
cracking could be sped up by knowledge of client \ccc queries, and one can
extend our results to consider the offline setting by increasing $\q$ to reflect
computational limits on adversaries (e.g., $\q = 10^{10}$) rather than limits on remote
login attempts.  Roughly speaking, we expect the leakage of \hibp and \gpc
to be proportionally as damaging here, and that our new protocol \fsbp will not provide as much
benefit for very large $\q$ (see discussion in \secref{sec:empirical}). \idbp will
provide no benefit to offline cracking attacks (assuming they already know the
username). 

Finally, we focus in threat model (2) on honest-but-curious adversaries, meaning that the malicious server
does not deviate from its protocol. Such actively malicious servers could lie to
the client about the contents of~$\breachDB$ in order to encourage them to pick
a weak password. Monitoring techniques might be useful to catch such
misdeeds. For the protocols we consider, we do not know of any other
active attacks advantageous to the adversary, and do not consider them further.  

\paragraph{Potential approaches.}
A \ccc protocol requires, at core, a secure set membership query. Existing
protocols for private set intersection (a generalization of set
membership)~\cite{kolesnikov2016efficient,pinkas2015phasing,pinkas2018efficient,chen2017fast} cannot currently scale to the set sizes required in \ccc
settings,  $N \approx 2^{30}$. For example, the basic PSI protocol that uses
an oblivious pseudorandom function (OPRF)~\cite{kolesnikov2016efficient} computes
$y_i = F_\key(\id_i,\pw_i)$ for $(\id_i,\pw_i) \in \breachDB$ where $F_\key$ is
the secure OPRF with secret key $\key$ (held by the server). It sends all
$y_1,\ldots,y_N$ to the client, and the client  obtains
$y = F_\key(\id,\pw)$ for its input $(\id,\pw)$ by obliviously computing it with
the server. The client can then check if $y \in \{y_1,\ldots,y_N\}$. But clearly
for large $N$ this is prohibitively expensive in terms of bandwidth.  One can
use Bloom filters to more compactly represent the set $y_1,\ldots,y_N$, but the
result is still too large. While more advanced PSI protocols exist that improve
on these results asymptotically, they are unfortunately not yet practical for
this \ccc setting~\cite{kolesnikov2016efficient,kiss2017private}.

Practical \ccc schemes therefore relax the security requirements, allowing
the protocol to leak some information about the client's queried $(\id,\pw)$ but
hopefully not too much. To date no one has investigated how damaging the leakage
of currently proposed schemes is, so we turn to doing that next. In \figref{fig:c3s-settings}, we show all the different settings for \ccc we discuss in the paper and compare their security and performance.  The security loss in \figref{fig:c3s-settings} is a comparison against an attacker that only has access to the username corresponding to a \ccc query (and not a bucket identifier).

\begin{figure}[t]
  \centering\footnotesize
  \begin{tabular}[t]{p{1.2cm}lp{2.4cm}rrr}\toprule
    Credentials checked & Name & Bucket identifier & \multicolumn{1}{p{.3cm}}{B/w (KB)} & \multicolumn{1}{p{.3cm}}{RTL (ms)} & \multicolumn{1}{p{0.8cm}}{Security loss}\\\midrule
    \multirow{2}{*}{Password}
                       & \hibp & $20$-bits of SHA1$(\w)$ & 32 & 220 & 12x\\[2pt]
                       & \ouralgo & \figref{fig:bucketize-hash-func}, $\qbar=10^2$ & 558 & 527 & 2x \\\midrule
    (Username, & \google & 16-bits of Argon2($u\Vert w$) & 1,066 & 489 &10x\\
    password) & \idbp & 16-bits of Argon2($\id$) & 1,066 & 517 & 1x \\

    \bottomrule
  \end{tabular}
  \caption{Comparison of different \ccc protocols and their bandwidth usage, round-trip latency, and security loss (compared to an attacker that has no bucket identifier information). \hibp~\cite{HIBP:v2} and
    \gpc~\cite{gpc-blog:2019} are two C3 services used in practice. We introduce
    frequency-smoothing bucketization ($\fsbp$) and identifier-based
    bucketization (\idbp). Security loss is computed assuming query budget
    $\q=10^3$ for users who have not been compromised before.}
  \label{fig:c3s-settings}
\end{figure}


%% file: preliminaries.tex
\section{Bucketization Schemes and Security Models}
\label{sec:prelims}

In this section we formalize the security models for a class of C3 schemes that bucketize the breach dataset into smaller sets (buckets). 
Intuitively, a 
straightforward approach for checking whether or not a client's credentials are
present in a large set of leaked credentials hosted by a server is to divide the leaked
data into various buckets.  The client and server can then perform a
private set intersection between the user's credentials and one of the buckets
(potentially) containing that credential. The bucketization makes private
set membership tractable, while only leaking to the server that the password may lie 
in the set associated to a certain bucket. 

We give a general framework to understand the security loss and bandwidth overhead of different bucketization schemes, and we will use this framework to evaluate existing C3 services. 

\paragraph{Notation.}
To easily describe our constructions, we fix some notation. 
Let $\PW$ be the set of all passwords, and $\distw$ be the associated
probability distribution; let $\users$ be the set of all user
identifiers, and $\dist$ be the joint distribution over $\users\times\PW$.  We will use
$\creds$ to denote the domain of credentials being checked, i.e., for
password-only C3 service, $\creds = \PW$, and for username-password C3 service,
$\creds = \users\times \PW$. Below we will use $\creds$ to give a generic
scheme, and specify the setting only if necessary to distinguish.  Similarly,
$\cred\in \creds$ denotes a password or a username-password pair,
based on the setting.
Let $\breachDB$ be the set of leaked credentials, and $|\breachDB| = N$.

Let 
$\H$ be a cryptographic hash function from
$\zo^*\mapsto\zo^\hlen$, where $\hlen$ is a parameter of the system. 
We use $\bucketset$ to denote the set of buckets, and we let
$\B\colon \creds \mapsto \powerset{\bucketset} \setminus\{\varnothing\}$ be a bucketizing function which maps a credential to a set
of buckets.  A credential can be mapped to multiple buckets, and every credential
is assigned to at least one bucket. An inverse function to $\B$ is 
$\A \colon \bucketset \mapsto \powerset{\creds}$, which maps a bucket to the set of all credentials
it contains; so, $\A(b) = \left\{\cred\in\creds\given b\in\B(\s) \right\}$. Note, $\A(b)$ can be very large given it considers all credentials in $\creds$. We let 
$\Abar$ be the function that denotes the credentials in the buckets held by the C3 server,
$\Abar(b) = \A(b)\cap \breachDB$.

The client sends $b$ to the server, and
then the client and the server engage in a set intersection protocol between
$\{\s\}$ and $\Abar(b)$. 

\begin{figure}[t]
  \centering\footnotesize
  \begin{tabular}[t]{lp{6cm}}
    \toprule
    Symbol & Description\\
    \midrule
    $\id\;/\;\users$ & user identifier (e.g., email) / domain of users \\
    $\pw\;/\;\PW$ & password / domain of passwords\\
    $\creds$ & domain of credentials \\
    $\breachDB$ & set of leaked credentials, $|\breachDB| = \dblen$ \\
    \hline
    $\dist$ & distribution of username-password pairs over $\users\times\PW$\\
    $\distw$ & distribution of passwords over $\PW$\\
    $\dists$ & estimate of $\distw$ used by C3 server\\\hline
    $\q$ & query budget of an attacker\\
    $\qbar$ & parameter to $\fsbp$, estimated query budget of an attack\\
    $\B$ & function that maps a credential to a set of buckets\\
    $\A$ & function that maps a bucket to the set of credentials it contains\\
    \bottomrule
  \end{tabular}
  \caption{The notation used in this paper.}
  \label{fig:symdef}
\end{figure}

\begin{figure}[t]
\medskip
  \centering
  \centering \hfpages{0.2}{
      $\underline{\Guessgzero^{\advA}(\q)}$\vspace{2pt}\\
      $(\user, \w) \genfrom{\dist} \users\times\PW$\\
      $\{\typo_1,\ldots,\typo_{\q}\} \gets \advA(\user, \q)$\\
      return $\w \in \{\typo_1,\ldots,\typo_{\q}\}$
    }{
    $\underline{\Guessgone_{\B}^{\advB}(\q)}$\vspace{2pt}\\
    $(\user, \w) \genfrom{\dist} \users\times\PW$;\;$\cred\gets (\user,\w)$\\
    $B \gets \B(\cred)$;\;$b \getsr B$\\
    $\{\typo_1,\ldots,\typo_{\q}\} \gets \advB(u, b, \q)$\\
    return $\w \in \{\typo_1,\ldots,\typo_{\q}\}$
    }  
    \caption{The guessing games used to evaluate security. }
  \label{fig:guess-games}
\end{figure}

\paragraph{Bucketization schemes.} Bucketization  divides the credentials
held by the server into smaller buckets. 
The client can use the bucketizing function $\B$ to find the set of buckets for
a credential, and then pick one randomly to query the server.  There are
different ways to bucketize the credentials.

In the first method, which we call hash-prefix-based bucketization (\hpbp), the
credentials are partitioned based on the first $\prefixlen$ bits of a
cryptographic hash of the credentials. \google~\cite{gpc-blog:2019} and
\hibp~\cite{HIBP:v2} APIs use \hpbp. The distribution of the credentials is not
considered in \hpbp, which causes it to incur higher security loss, as we show
in~\secref{sec:hpb}.

We introduce a new bucketizing method, which we call frequency-smoothing
bucketization (\ouralgo), that takes into account the distribution of the
credentials and replicates credentials into multiple buckets if necessary. The
replication ``flattens'' the conditional distribution of passwords given a
bucket identifier, and therefore vastly reduces the security loss. We discuss \ouralgo
in more detail in \secref{sec:fsb}.

In both \hpbp and \ouralgo, the bucketization function depends on the user's
password.  We give another bucketization approach --- the most secure one
--- that bucketizes based only on the hash prefix of the user identifier. We call
this identifier-based bucketization (\idbp). This approach is only applicable for
username-password C3 services. We discuss \idbp in \secref{sec:hpb}.

\paragraph{Security measure.}
The goal of an attacker is to learn the user's password. We will focus on
online-guessing attacks, where an attacker tries to guess a user's password over
the login interface provided by a web
service. 
An account might be locked after too many incorrect guesses (e.g.,
$10$), in which case the attack fails.  Therefore, we will measure an attacker's success
given a certain guessing budget $\q$.  
We will always assume the attacker has access to the username of
the target user.

The security games are given in~\figref{fig:guess-games}.  The game $\Guessgzero$
models the situation in which no information besides the username is revealed to
the adversary about the password. In the game $\Guessgone$, the adversary also gets
access to a bucket that is chosen according to the credentials $\cred=(\user,\w)$ and the
bucketization function $\B$.

We define the advantage against a game as the maximum probability that the game
outputs 1. Therefore, we maximize the probability, over all adversaries, of the adversary winning the game in $q$ guesses. 
\bnm\advzero(\q) = \max_{\advA}\ \Prob{\Guessgzero^\advA_{}(q) \Rightarrow 1}\ ,\enm and
\bnm\advone{\B}(\q) = \max_{\advB}\ \Prob{\Guessgone^\advB_{\B}(q) \Rightarrow 1}\ .\enm The probabilities are taken over the choices
of username-password pairs and the selection of bucket via the bucketizing function $\B$.
The security loss $\secloss_\B(\q)$ of a bucketizing protocol $\B$ is defined as 
\bnm
\secloss_\B(\q)=\advone{\B}(\q) -\advzero(\q)\,. 
\enm

Note, 
\bnm
\Prob{\Guessgzero^\advA_{}(q) \Rightarrow 1} = \sum_\user \Prob{\w\in\advA(u, q) \wedge U=u } \;.\enm
To maximize this probability, the attacker must pick the $\q$ most probable
passwords for each user. Therefore, %
\bne\advzero(q) = \sum_\user \max_{\wss{q}} \sum_{i=1}^q
\Prob{W=\w_i\wedge U=\user}\;.
\label{eq:adv-zero}\ene
\noindent In $\Guessgone_{\B}$, the attacker has access to the bucket
identifier, and therefore the advantage is computed as %
\begin{align}
\advone{\B}(q) 
&= \sum_{\user}\sum_{b}
\max_{\wss{\q}} \sum_{i=1}^q \Prob{W=\w_i\wedge U=\user \wedge
	B=b}\nonumber\\
&=\sum_{\user}\sum_{b}
\max_{\substack{\uwss{\q}\\\in\A(b)}} \sum_{i=1}^q \frac{\Prob{W=\w_i\wedge U=\user}}{|\B((\user, \w_i))|}          
\label{eq:adv-one}\end{align}
The second equation follows because for $b\in\B((\user,\w))$, each bucket in $\B(w)$ is equally
likely to be chosen, so\bnm\CondProb{B=b}{W=\w\wedge U=u} = \frac{1}{|\B((\user, \w))|}\ .\enm

The joint distribution of usernames and passwords is hard to model. To simplify the
equations, we divide the users targeted by the attacker into two groups:
\emph{compromised} (users whose previously compromised accounts are available to
the attacker) and \emph{uncompromised} (users for which the attacker has no
information other than their usernames).  

We assume there is no direct correlation between the username and
password.\footnote{Though prior work~\cite{wang2016targeted,li2016study}
  suggests knowledge of the username can improve efficacy of guessing 
  passwords, the improvement is minimal.  See~\appref{app:user-pw-correlation}
  for more on this.}  Therefore, an attacker cannot use the knowledge
of only the username to tailor guesses.  This means that in the uncompromised
setting, we assume $\CondProb{W=\w}{U=u} = \Prob{W=\w}$.  Assuming independence
of usernames and passwords, we define in the uncompromised setting
\bne\label{eq:lambda}
\lambda_{\q}=\advzero(q)=\max_{\wss{\q}}\sum_{i=1}^{\q}\Prob{W=\w_i}\,.  \ene

We give analytical (using Equations~\ref{eq:adv-one} and \ref{eq:lambda}) and empirical analysis
of security in this setting, and show that the security of uncompromised users is impacted by existing C3 schemes much more than that of compromised users.

In the compromised setting, the attacker can use the username to find
other leaked passwords associated with that user, which then can be used to
tailor guesses~\cite{pal2019beyond,wang2016targeted}. Analytical bounds on the compromised setting (using Equations~\ref{eq:adv-zero} and \ref{eq:adv-one}) are less informative, so we evaluate this setting empirically in~\secref{sec:empirical}.   

\paragraph{Bandwidth.} The bandwidth required for a bucketization scheme is
determined by the size of the buckets. 
The maximum size of the buckets can be
determined using a balls-and-bins approach~\cite{Berenbrink2008511}, assuming the client picks a bucket randomly from the possible
set of buckets $\B(s)$ for a credential $s$, and $\B(s)$ also maps $s$ to a random set of buckets. 
 In total
$m=\sum_{s\in\breachDB} |\B(s)|$ credentials (balls) are ``thrown'' into $n=\nbuckets$
buckets.  
If $m > \nbuckets\cdot\log\nbuckets$, then standard results
~\cite{Berenbrink2008511}
give that the maximum number of passwords in a bucket  is less
than
$\frac{m}{n}\cdot\left(1 + \sqrt{\frac{n\log n}{m}}\right) \le
2\cdot\frac{m}{n}$, with very high probability $1 - o(1)$.  We will use this formula to compute an upper bound on the bandwidth
requirement for specific bucketization schemes.

For $\hpbp$ schemes, each credential will be mapped to a random bucket if we assume that the hash function acts as a random oracle. For $\fsbp$, since we only randomly choose the first bucket and map a credential to a range of buckets starting with the first one, it is not clear that the set of buckets a credential is mapped to is random. We also show empirically that these bounds hold for the \ccc schemes. 



%% file: HIBP.tex
\section{Hash-prefix-based Bucketization}
\label{sec:hpb}

Hash-prefix-based bucketization ($\hpbp$) schemes are a simple way to divide the
credentials stored by the \ccc server.  In this type of \ccc scheme, a prefix of the hash of
the credential is used as the criteria to group the credentials into buckets ---
all credentials that share the same hash-prefix are assigned to the same bucket.
The total number of buckets depends on $\prefixlen$, the length of the
hash-prefix. The number of credentials in the buckets depends on both~$\prefixlen$ and
$|\breachDB|$.
We will use $\Htl(\cdot)$ to denote the function that outputs the $\prefixlen$-bit prefix of the hash $\H(\cdot)$.  The client shares the hash prefix of the credential they wish to check with the
server. While a smaller hash prefix reveals less information to the server about
the user's password, it also increases the size of each bucket held by the
server, which in turn increases the bandwidth overhead. 

Hash-prefix-based bucketization is currently being used for credential checking in industry
: \hibp~\cite{HIBP:v2} and
\google~\cite{gpc-blog:2019}. We introduce a new \hpbp protocol called \idbp
that achieves zero security loss for any query budget. Below we will discuss the
design details of these three C3 protocols.

\paragraph{\hibp~\cite{HIBP:v2}.}  \hibp uses \hpbp bucketization to
provide a password-only C3 service.  They do not provide compromised username-password checking.
\hibp maintains a database of leaked passwords, which contains more than 501
million passwords~\cite{HIBP:v2}. 
They use the SHA1 hash function, with prefix length $\prefixlen = 20$; the leaked
dataset is \emph{partitioned} into $2^{20}$ buckets.  The prefix length is
chosen to ensure no bucket is too small or too big. With $\prefixlen=20$,
the smallest bucket has 381 passwords, and  the largest bucket has 584 passwords~\cite{leakcheck-kanon} .  This effectively makes
the user's password $k$-anonymous. However, $k$-anonymity provides limited protection,
as shown by numerous prior
works~\cite{naranyanan2008robust,machanavajjhala2006diversity,zhang2007information}
and by our security evaluation. 
 
The passwords are hashed using SHA1 and indexed by their hash prefix for
fast retrieval.  A client computes the SHA1 hash of their password $\w$ and
queries \hibp with the $20$-bit prefix of the hash; the server responds with all
the hashes that share the same 20-bit prefix. The client then checks if the
full SHA1 hash of $\w$ is present among the set of hashes sent by the
server. This is a weak form of PSI that does not hide the leaked passwords from
the client --- the client learns the SHA1 hash of the leaked passwords and can
perform brute force cracking to recover those passwords.

\hibp justifies this design choice by observing that passwords in the server
side leaked dataset are publicly available for download on the Internet.
Therefore, \hibp lets anyone download the hashed passwords and usernames. This
can be useful for parties who want to host their own leak checking service
without relying on HIBP.  However, keeping the leaked dataset up-to-date can be
challenging, making a third-party C3 service  preferable.

\hibp trades server side privacy for protocol simplicity.  The protocol
also allows utilization of caching on content delivery networks (CDN),
such as Cloudflare.\footnote{\url{https://www.cloudflare.com/}}  Caching
helps \hibp to be able to serve 8 million requests a day with 99\%
cache hit rate (as of August 2018)~\cite{hibp:cache}. 
The human-chosen password distribution is ``heavy-headed'',
that is a small number of passwords are chosen by a large number of
users. Therefore, a small number of passwords are queried a large number of
times, which in turn makes CDN caching much more
effective.  

\paragraph{\google~\cite{gpc-blog:2019, thomas2019protecting}.}  
Google provides a username-password \ccc service, called Password Checkup (\google).
The client --- a browser extension --- computes the hash of the username and
password together using the Argon2 hash function (configured to use a single thread, 256 MB of memory, and a time cost of three) with the first $\prefixlen=16$
bits to determine the bucket identifier. After determining the bucket, the client
engages in a private set intersection (PSI) protocol with the server. The full
algorithm is given in~\figref{fig:google}. \google uses a computationally expensive hash function to make it more difficult for an adversary to make a large number of queries to the server. 

\google uses an OPRF-based PSI protocol~\cite{thomas2019protecting}.  Let $F_{a}(x)$ be a function that first calls the hash function $\H$ on $x$, then maps the hash output onto an elliptic curve point, and finally, exponentiates the elliptic curve point (using elliptic curve group operations) to the power $a$. Therefore it holds that $(F_{a}(x))^b=F_{ab}(x)$. 

The server has a secret key $\key$ which it uses to compute the values 
$y_i = F_\key(\id_i\|\pw_i)$ for each $(\id_i,\pw_i)$ pair in the breach dataset.  The client shares with the server the bucket identifier~$b$ and the PRF output $x =F_r(\id\|\w)$, for some
randomly sampled $r$. The server returns the bucket $\zbf_b = \{y_i\given \H(\id_i\|\w_i) = b\}$
and $y=x^{\key}$. Finally, the client completes the OPRF computation by computing $\tilde{x} = y^{\frac{1}{r}} = F_\key(\id\|\w)$, and checking if $\tilde{x} \in \zbf_b$.

The \google protocol is significantly more complex than \hibp, and it does not
allow easy caching by CDNs. However, it provides secrecy of server-side leaked
data --- the best case attack is to follow the protocol to brute-force check if
a password is present in the leak database.

\begin{figure}[t]
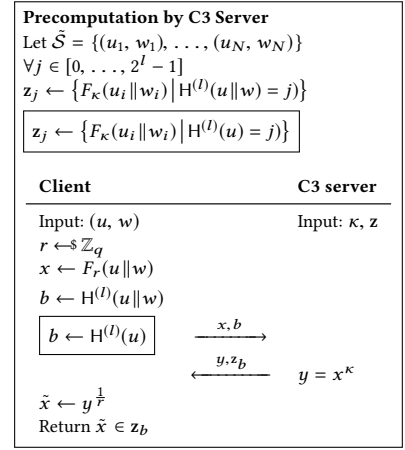

  \centering\footnotesize
  \fpage{0.277}{\footnotesize
    \textbf{Precomputation by C3 Server}\\
    Let $\breachDB = \{(\user_1,\w_1), \ldots, (\user_\dblen,\w_\dblen)\}$\\
    $\forall j\in [0,\ldots, {2^\prefixlen-1}]$\\
    $\zbf_j\gets\left\{ F_\key(\id_i\|\w_i)\given \Htl(\id\|\w) =j)\right\}$\\[2pt]
    \fbox{${\zbf_j\gets\left\{ F_\key(\id_i\|\w_i)\given  \Htl(\id) =j)\right\}}$}\\[5pt]
  \[
\begin{array}{l c l}\footnotesize
\textbf{Client} &  & \textbf{\ccc server}\\\midrule
\textnormal{Input: } (\id,\w) & & \textnormal{Input: } \key, \zbf\\
  r \getsr \Z_q & & \\
  x \gets F_r(\user\|\pw) &&\\[2pt]
  b \leftarrow  \Htl(\user\|\pw) & & \\[2pt]
  \fbox{$b  \leftarrow \Htl(\user)$} & \xrightarrow{\hspace{1em}x,b\hspace{1em}} & \\
                & \xleftarrow{\hspace{1em} y,\zbf_b \hspace{1em}}& y = x^{\key} \\
  \tilde{x} \leftarrow y^{\frac{1}{r}}& & \\
  \text{Return } \tilde{x}\in \zbf_b & & 
\end{array}
\]
}
  \caption{Algorithms for \google, and the change in \idbp given in the box. $F_{(\cdot)}(\cdot)$ is  a PRF. }
  \label{fig:google}
\end{figure}

\paragraph{Bandwidth.} \hpbp assigns each credential to only one bucket;
therefore, $m = \sum_{\w\in\breachDB}|\B(\w)| = |\breachDB| = N$. The total
number of buckets is $n=2^\prefixlen$. Following the discussion
from~\secref{sec:prelims}, the maximum bandwidth for a \hpbp C3S should be no more than
$2\cdot\frac{m}{n} = 2\cdot\frac{N}{2^\prefixlen}$.
 
We experimentally verified bandwidth usage, and the 
sizes of the buckets for \hibp, \google, and \idbp are given in~\secref{sec:perf}.

\paragraph{Security.} 
\hpbp schemes like \hibp and \google expose a prefix of the user's password (or username-password pair) to the
server.  As discussed earlier, we assume the attacker knows the username of the
target user. In the uncompromised setting --- where the user identifier does not
appear in the leaked data available to the attacker, we show that giving the attacker
the hash-prefix with a guessing budget of $\q$ queries is equivalent to giving as many as $\q\cdot\nbuckets$ queries (with no hash-prefix) to the attacker. As a reminder, $\nbuckets$ is the number of buckets. For example, consider a \ccc scheme that uses a 5-character hash prefix as a bucket identifier ($2^{20}$ buckets). If an attacker has 10 guesses to figure out a password, then given a bucket identifier, they can eliminate any guesses on their list that don't belong in that bucket. If their original guesses are distributed equally across all buckets, then knowing the 5-character hash prefix can help them get through around $\q\cdot2^{20}$ of those guesses. 
\begin{theorem}
  Let $\B_{\textnormal{\hpbp}}:\creds\mapsto\bucketset$ be the bucketization
  scheme that, for a credential $s\in\creds$, chooses a bucket that is a
  function of $\Htl(s)$, where $s$ contains the user's password.  The advantage
  of an attacker in this setting against previously uncompromised users is
  \bnm\advone{\B_{\textnormal{\hpbp}}}(\q) \le \advzero(\q\cdot\nbuckets)\;.\enm
  \label{th:hpb}\vspace{-1em}
\end{theorem}
\begin{proof}
First, note that $|\B_{\hpbp}(s)|=1$, for any input $s$, as every password is assigned to exactly one of the buckets. Following the discussion from~\secref{sec:prelims}, assuming independence of usernames and passwords in the uncompromised setting, we can compute
the advantage against game $\Guessgone$ as 
\begin{align*}
  \advone{\B_{\hpbp}}(\q) 
  &= \sum_{b\in\bucketset} \max_{\substack{\wss{\q}\\\in\A(b)}}\sum_{i=1}^{\q}\Prob{W=\w_i}
  \le \advzero(\q\cdot\nbuckets)\,.
\end{align*}
We relax the $\A(b)$ notation to denote set of passwords (instead of
username-password pairs) assigned to a bucket $b$. The inequality follows from the fact that each password is present in only one bucket. If we sum up the probabilities of the top $q$ passwords in each bucket, the result will be at most the sum of the probabilities of the top $q\cdot \nbuckets$ passwords. Therefore, the maximum advantage achievable is 
$\advzero(\q\cdot\nbuckets)$. \qed
\end{proof}

\thref{th:hpb} only provides an upper bound on the security loss. Moreover, for
the compromised setting, the analytical formula in \eqref{eq:adv-one} is not very informative. So, we use
empiricism to find the effective security loss against compromised and
uncompromised users. We report all security simulation results
in~\secref{sec:empirical}.  
Notably, with \gpc using a hash prefix length $\prefixlen=16$, an attacker can
guess passwords of  59.7\% of (previously uncompromised) user accounts 
in fewer than 1,000 guesses, over a 10x increase from the
percent it can compromise without access to the hash
prefix. (See~\secref{sec:empirical} for more results.)

\paragraph{Identifier-based bucketization (\idbp).}
As our security analysis and simulation show, the security degradation of \hpbp
can be high. The main issue with those protocols is that the bucket identifier is
a deterministic function of the user password. We give a new C3 protocol that
uses \hpbp style bucketing, but based only on username.  We call this
identifier-based bucketization (\idbp).  \idbp is defined for username-password
C3 schemes.

\idbp is a slight modification of the protocol used by \google --- we use
the hash-prefix of the username, $\Htl(u)$, instead of the hash-prefix of the
username-password combination, $\Htl(u\concat w)$, as a bucket identifier. The
scheme is described in~\figref{fig:google}, using the changes in the boxed code.
The bucket identifier is computed completely independently of the password
(assuming the username is independent of the password). Therefore, the attacker gets
no additional advantage by knowing the bucket identifier.

Because \idbp uses the hash-prefix of the username as the bucket identifier, two
hash computations are required on the client side for each query (as opposed to
one for \google). With most modern devices, this is not a significant computing
burden, but the protocol latency may be impacted, since we use a slow hash
(Argon2) for hashing both the username and the password. We show experimentally how the extra hash computation affects
the latency of \idbp in \secref{sec:perf}.

Since in \idbp, the bucket identifier does not depend on the user's password,
the conditional probability of the password given the bucket identifier remains
the same as the probability without knowing the bucket identifier. As a result, exposing the bucket identifier does not lead to security loss.
\begin{theorem}
  With the  \textnormal{\idbp} protocol, for  all $q\geq 0$
  \bnm\advone{\textnormal{\idbp}}(\q) = \advzero(\q).\enm
  \label{thm:idbp}
  \vspace{-1em}
\end{theorem}
\begin{proof}
	Because the \idbp bucketization scheme does not depend on the password,
	$\CondProb{B=b}{W=w\wedge U=u} = \CondProb{B=b}{ U=u}$.
	
	We can upper bound the success rate of an adversary in the $\Guessgone_{\idbp}$ game by
	\begin{align*}
	&\advone{\idbp}(\q) \\
	&= \sum_\user\sum_{b}
	\max_{\wss{\q}} \sum_{i=1}^q \Prob{W=\w_i\wedge U=\user}\cdot\CondProb{B=b}{U=\user}\\
	&= \sum_\user\left(\sum_{b} \CondProb{B=b}{U=\user}\right)
	\max_{\wss{\q}} \sum_{i=1}^q \Prob{W=\w_i\wedge U=\user}\\
	&=\advzero(\q)
	\end{align*}
	The first step
	follows from independence of password and bucket choice, and the third step is
	true because there is only one bucket for each username.
	\qed
\end{proof}

We would like to note, though \idbp reveals nothing about the password, learning
the username becomes easier (compared to \gpc) --- an attacker can narrow down
the potential users after seeing the bucket identifier. While this can be
concerning for user's privacy, we believe the benefit of not revealing anything about the
user's password outweighs the risk.

Unfortunately, \idbp does not work for the password-only \ccc setting because it requires that the server store username-password pairs. In the next section we 
introduce a more secure password-only \ccc scheme. 


%% file: FSBP.tex
\section{Frequency-Smoothing Bucketization}\label{sec:fsb}
In the previous section we showed how to build a username-password \ccc service that
does not degrade security. However, many services, such as $\hibp$, only
provide a password-only \ccc service. HIBP does not store username-password pairs
so, should the \hibp server ever get compromised, an attacker cannot use their
leak database to mount credential stuffing attacks. 
Unfortunately, \idbp cannot
be extended in any useful way to protect password-only \ccc services.

Therefore, we introduce a new bucketization scheme to build secure password-only
\ccc services. We call this scheme frequency-smoothing bucketization (\fsbp). \fsbp
assigns a password to multiple buckets based on its probability --- frequent
passwords are assigned to many buckets.  Replicating a password into multiple
buckets effectively reduces the conditional probabilities of that password
given a bucket identifier. We do so in a way that makes the conditional
probabilities of popular passwords similar to those of unpopular passwords to make it harder for the attacker to guess the correct password. 
\ouralgo, however, is only effective for non-uniform credential distributions,
such as password distributions.\footnote{Usernames (e.g., emails) are unique for
  each users, so the distribution of usernames and username-password
  pairs are close to uniform.} Therefore, \fsbp cannot be used to build a
username-password \ccc
service.

Implementing \ouralgo requires knowledge of the distribution of human-chosen
passwords. Of course, obtaining precise knowledge of the password distribution
can be difficult; therefore, we will use an estimated password distribution, denoted by $\dists$. 
Another parameter of $\ouralgo$ is $\qbar$, which is an estimate of the attacker's query budget. We show that if the actual query budget $\q\le \qbar$, \ouralgo
has zero security loss. Larger $\qbar$ will provide better security; however, it
also means more replication of the passwords and larger bucket
sizes. So, $\qbar$ can be tuned to balance between security and bandwidth. Below
we will give the two main algorithms of the \ouralgo scheme: $\Bfsb$ and $\Abarfsb$,
 followed by a bandwidth and security analysis.

\paragraph{Bucketizing function ($\Bfsb$).}
To map passwords to buckets, we use a hash function
$f:\PW\mapsto\Z_{\nbuckets}$. The algorithm for bucketization $\B_\ouralgo(\w)$ is given in
\figref{fig:bucketize-hash-func}. The parameter $\qbar$ is used
in the following way:
 $\B$ replicates the most probable $\qbar$
passwords, $\PW_{\qbar}$, across all $\nbuckets$ buckets. Each of the remaining
passwords are replicated proportional to their probability.  A password $\w$
with probability $\dists(\w)$ is replicated exactly
$\bw=\ceil{\frac{\nbuckets\cdot \dists(\w)}{\qdist}}$ times, where
$\w_{\qbar}$ is the $\qbar\thh$ most likely
password. 
Exactly which buckets a password is assigned to are determined using the hash
function $f$.  Each bucket is assigned an identifier between $[0, \nbuckets-1]$.
A password $\w$ is assigned to the buckets whose identifiers fall in the range
$\left[f(\w), f(\w) + \bw-1\right]$. The range can wrap around. For
example, if $f(\w) + \bw > \nbuckets$, then the password is assigned to the
buckets in the range $[0, f(\w) + \bw -1 \mod \nbuckets]$ and
$[f(\w), |\bucketset|-1]$.

\begin{figure}[t]
  \centering 
    \hfpagess{0.2}{0.2}{\footnotesize
    $\underline{\Bfsb(\w):}$\\[2pt]
    $\bw \gets \min\left\{\nbuckets, \ceil{\frac{\nbuckets\cdot \dists(\w)}{\qdist}}\right\}$\\
    $s \gets f(\w)$\\
    If $s + \bw < \nbuckets$ then \\
    \myind $\;r\gets [s, s+\bw-1]$\\
    Else \\
    \myind $\;r\gets [0, s+\bw -1 \mod \nbuckets]$\\
    \myind $\;r\gets r\cup[s, \nbuckets-1]$\\
    Return $r$
  }
  {\footnotesize
      $\underline{\Abarfsb(b):}$\\
      /* returns $\{\w\in\breachDB \given b\in \B(\w)\}$ */\\
      $ A \gets \PW_{\bar{q}}$\\
      For $\w \in \breachDB\setminus \PW_{\bar{q}}$ do\\
      \mytab If $b\in\Bfsb(w)$ then \\
      \mytab \mytab $A\gets A\cup \{\w\}$\\
      return $A$
  }
  \caption{Bucketizing function $\Bfsb$ for assigning passwords to buckets in \ouralgo.
    Here $\dists$ is the distribution of passwords; $\PW_{\qbar}$ is the set of
    top-${\qbar}$ passwords according to $\dists$; $\bucketset$ is the set of buckets;
    $f$ is a  hash function $f\colon W\mapsto \Z_{|B|}$; $\breachDB$ is the set
    of passwords hosted by the server.  }
    \label{fig:bucketize-hash-func}
\end{figure}

\paragraph{Bucket retrieving function ($\Abar$).}
Retrieving passwords assigned to a bucket is challenging in $\fsbp$. An
inefficient --- linear in $\dblen$ --- implementation of $\Abar$ is given
in~\figref{fig:bucketize-hash-func}. Storing the contents of each bucket
separately is not feasible, since the number of buckets in $\ouralgo$ can be
very large, $\nbuckets \approx \dblen$.
To solve the problem, we utilize the structure of the bucketizing procedure where
passwords are assigned to buckets in continuous intervals.  This allows us to use an 
interval tree~\cite{wiki:intervaltrees} data structure to store the intervals
for all of the passwords. Interval trees allow fast queries to retrieve the set of
intervals that contain a queried point (or interval) --- exactly what is needed
to instantiate~$\Abar$.

This efficiency comes with increased storage cost: storing $\dblen$
entries in an interval tree requires $\bigO{\dblen\log\dblen}$ storage.  The
tree can be built in $\bigO{\dblen\log\dblen}$ time, and each query
takes 
$\bigO{\log \dblen + |\Abar(b)|}$ time. 
The big-O notation only hides small constants.

\paragraph{Estimating password distributions.} To construct the bucketization
algorithm for \ouralgo, the server needs an estimate of the password
distribution $\distw$.  This estimate will be used by both the server and the
client to assign passwords to buckets. One possible estimate is the
histogram of the passwords in the leaked data~$\breachDB$. Histogram estimates
are typically accurate for popular passwords, but such estimates are not complete
--- passwords that are not in the leaked dataset will have zero probability
according to this estimate. Moreover, sending the histogram over to the client
is expensive in terms of bandwidth, and it may leak too much information about the dataset.  We also considered password
strength meters, such as \zxcvbn~\cite{wheeler2016zxcvbn} as a proxy for a 
probability estimate. However, this estimate turned out to be too coarse for our purposes. For
example, more than $10^5$ passwords had a ``probability'' of greater than $10^{-3}$.

We build a $3$-gram password model $\distn$ using the leaked passwords present
in $\breachDB$.  Markov models or $n$-gram models are shown to be effective at
estimating human-chosen password distributions~\cite{ma2014study}, and they are
very fast to train and run (unlike neural network based password distribution
estimators, such as~\cite{melicher2016fast}).  However, we found the $n$-gram
model assigns very low probabilities to popular passwords. The sum of the
probabilities of the top 1,000 passwords as estimated by the 3-gram model is
only 0.032, whereas those top 1,000 passwords are chosen by $6.5\%$ of users.

We therefore use a combined approach that uses a histogram model for the popular
passwords and the 3-gram model for the rest of the distribution. Such combined
techniques are also used in practice for password strength
estimation~\cite{wheeler2016zxcvbn,melicher2016fast}. Let $\dists$ be the
estimated password distribution used by \ouralgo. Let $\distbreach$ be the
distribution of passwords implied by the histogram of passwords present in
$\breachDB$. %
Let $\breachDB_t$ be the set of the $t$ most probable passwords according to
$\distbreach$. We used $t=10^6$. Then, the final estimate is
 \bnm \dists(\w) = \left\{
  \begin{array}{ll}
    \distbreach(\w)& \text{if } \w\in\breachDB_t\ , \\
    \distn(w)\cdot\frac{1-\sum_{\typo\in\breachDB_t}\distbreach(\w)}{1-\sum_{\typo\in\breachDB_t}\distn(\w)} & \text{otherwise.}
  \end{array}
  \right.
\enm
Note that instead of using the 3-gram probabilities directly, we multiply them by a normalization factor that allows $\sum_w \distest(w)=1$, assuming that the same is true for the distributions $\distbreach$ and $\distn$.

\paragraph{Bandwidth.}  We use the formulation provided in~\secref{sec:prelims}
to compute the bandwidth requirement for $\ouralgo$. In this case,
$m = \nbuckets\cdot\qbar + \frac{\nbuckets}{\qdist} + {\dblen}$, and
$n=\nbuckets$.  Therefore, the maximum size of a bucket is with high probability
less than
$2\cdot\left({\qbar} + \frac{1}{\qdist} + \frac{\dblen}{\nbuckets}\right)$.  The
details of this analysis are given in~\appref{app:fsbp}.

In practice, we can choose the number of buckets to be such that
$|\bucketset| = \dblen$. Then, the number of passwords in a bucket depends primarily on
the parameter ${\qbar}$. Note, bucket size increases with $\qbar$. 

\paragraph{Security analysis.}
We show that there is no security loss in the uncompromised setting
for \ouralgo when the actual
number of guesses $\q$ is less than the parameter $\qbar$ and the estimate $\distest$ is accurate. We also give a
bound for the security loss when $\q$ exceeds $\qbar$.

\begin{theorem}\label{thm:bucket-entropy}
  Let $\ouralgo$ be a frequency based bucketization scheme that ensures 
  $\forall \w \in \PW,\;$ $|\Bfsb(w)| = \min\left\{|\bucketset|,\,
     \ceil{\frac{|\bucketset|\cdot\dists(\w)}{\dists(\w_{\bar{q}})}}\right\}$.
   Assuming that the distribution estimate $\dists=\distw$, then for the uncompromised users, 
	\begin{newenum}
        \item 
          $\advone{\Bfsb}(\q) = \advzero(\q)\;\;$
          for $\q\le\qbar$, and\label{claim:nosecloss} 
          \item 
            for $\q > \qbar\;$,
            \bnm\frac{\lambda_{q} - \lambda_{{\bar{q}}}}{2} \le\; \secloss_{q} \; \leq (q-{\bar{q}})\cdot\qdist - (\lambda_{q}-\lambda_{\bar{q}}) \enm
            \label{claim:upper-lower-bounds}
	\end{newenum}\vspace{-1em}
\end{theorem}
Recall that the  probabilities $\lambda_q$ are defined  in \eqref{eq:lambda}.
We include the full proof for \thref{thm:bucket-entropy}  in~\appref{sec:bucket-entropy-proof}. 
Intuitively, since the top $\q$ passwords are repeated across all buckets, having a bucket identifier
does not allow an attacker to more easily guess these $\q$ passwords.
Moreover, the conditional probability of these $\q$ passwords given the bucket is greater than 
that of any other password in the bucket. Therefore, the attacker's best
choice is to guess the top $\q$ passwords, meaning that it does not get any
additional advantage when $\q\le\qbar$, leading to part
(\ref{claim:nosecloss}) of the theorem.

The proof of part (\ref{claim:upper-lower-bounds}) follows from the upper and
lower bounds on the number of buckets each password beyond the top ${\q}$ is
placed within. The bounds we prove show that the additional advantage in guessing
the password in $\q$ queries is less than the number of additional queries times
the probability of the $\qbar\thh$ password and at least half the difference in
the guessing probabilities $\lambda_\q$ and $\lambda_{\qbar}$.

Note that this analysis of security loss is based on the assumption that the \fsbp
scheme has access to the precise password distribution, $\dists=\distw$. We empirically analyze the security loss in~\secref{sec:empirical} for $\dists\neq \distw$, in both the compromised and uncompromised settings. 


%% file: simulation.tex
\section{Empirical Security Evaluation}
\label{sec:empirical}

In this section we empirically evaluate and compare the security loss
for different password-only \ccc schemes we have discussed so far --- hash-prefix-based
bucketization (\hpbp) and frequency-smoothing bucketization (\ouralgo). 

We focus on known-username attacks (KUA), since in many deployment settings a
curious (or compromised) \ccc server can figure out the username of the querying
user. 
We separate our analysis into two settings: previously \emph{compromised} users,
where the attacker has access to one or more existing passwords of the target
user, and previously \emph{uncompromised} users, where no password
corresponding to the user is known to the attacker (or present in the breached data).

We also focus on what the honest-but-curious \ccc server can learn from knowing the bucket. In our experiment, we show the success rate of an adversary that knows the exact leak dataset used by the server. We expect that an adversary that doesn't know the exact leak dataset will have slightly lower success rates. 

First we will look into the unrestricted setting where no password policy is
enforced, and the attacker and the \ccc server have the same amount of information
about the password distribution. In the second experiment, we analyze
the effect on security of giving the attacker more information compared to the \ccc server
(defender) by having a password policy that the attacker is aware of but the \ccc
server is not.

\begin{figure}[t]
  \footnotesize
  \begin{tabular}{lr|rr|rr}\toprule
    & $\breachDB$ & $\testdata$ & $\testdata\cap\breachDB$ & $\sptestdata$ & $\sptestdata\cap\breachDB$\\\midrule
    \# users    & 901.4 & 12.9 & 5.9 (46\%) & 8.4 & 3.9 (46\%)\\
    \# passwords& 435.9 & 8.9 & 5.7 (64\%) & 6.7 & 3.9 (59\%)\\
    \# user-pw pairs & 1,316.6 & 13.1 & 3.2 (24\%) & 8.5 & 2.0 (23\%)\\
    \bottomrule
  \end{tabular}
  \caption{Number of entries (in millions) in the breach dataset $\breachDB$,
    test dataset $\testdata$, and the site-policy test subset $\sptestdata$. Also reported  are
    the intersections (of users, passwords, and user-password pairs, separately)
    between the test dataset entries and the whole breach dataset that the
    attacker has access to. The percentage values refer to the fraction of the
    values in each test set that also appear in the intersections. }
  \label{fig:dataset-stats}
  
\end{figure}

\paragraph{Password breach dataset.} We used the same breach dataset used
in~\cite{pal2019beyond}. 
The dataset was derived from a previous breach compilation~\cite{leakurl}
dataset containing about $1.4$~billion username-password pairs. We chose to use this dataset rather than, for example, the password breach dataset from \hibp, because it contains username-password pairs.

We cleaned the data by removing non-ASCII characters and passwords longer than 30
characters. We also combined username-password pairs with the same case-insensitive username, and we removed users with over 1,000 passwords, as they didn't seem to be associated to real accounts. The authors of~\cite{pal2019beyond} also joined accounts with
similar usernames and passwords using a method they called the \emph{mixed
  method}. 
We joined the dataset using the same mixed method, but we also kept the usernames with only one email and password.

The final dataset consists of about 1.32 billion username-password
pairs.\footnote{Note, there are duplicate username-password pairs in this
  dataset. } We
remove $1\%$ of username-password pairs to use as test data, denoted as
$\testdata$.  The remaining $99\%$ of the data is used to simulate the database of
leaked credentials $\breachDB$.  For the experiments with an enforced password
policy, we took the username-password pairs in $\testdata$ that met the requirements of the password
policy to create $\sptestdata$. We use $\sptestdata$ to simulate queries from a website which
only allows passwords that are at least 8 characters long and are not present in Twitter's list of banned passwords~\cite{twitterbanned}.
For all attack simulations, the target user-password pairs are sampled from the test dataset
$\testdata$ (or $\sptestdata$).

In~\figref{fig:dataset-stats}, we report some statistics about $\testdata$,
$\sptestdata$, and $\breachDB$.  Notably, 5.9 million (46\%) of the users
in $\testdata$ are also present in $\breachDB$. 
Among the username-password pairs, 3.2 million (24\%) of the pairs in $\testdata$ are
also present in $\breachDB$. This means an attacker will be able to compromise
about half of the previously compromised accounts trivially with credential stuffing.  In the
site-policy enforced test data $\sptestdata$, a similar proportion of the users (46\%) and username-password pairs (23\%) are also present in $\breachDB$.

\newcommand{\testdatacomp}{\testdata_\textsf{comp}}
\newcommand{\testdatauncomp}{\testdata_\textsf{uncomp}}
\newcommand{\guessest}{\guesses_t}
\newcommand{\guessesh}{\guesses_h}
\newcommand{\guessesn}{\guesses_n}

\paragraph{Experiment setup.} 
We want to understand the impact of revealing a bucket identifier on the security of uncompromised and compromised
users separately.  As we can see from~\figref{fig:dataset-stats}, a large
proportion of users in $\testdata$ are also present in $\breachDB$. We therefore
split $\testdata$ into two parts: one with only username-password pairs from compromised users (users with at least one password present in $\breachDB$), $\testdatacomp$, and another
with only pairs from uncompromised users (users with no passwords present in $\breachDB$), $\testdatauncomp$. We generate two sets of random
samples of 5,000 username-password pairs, one from $\testdatacomp$, and another
from $\testdatauncomp$. We chose 5,000 because this number of samples led to a low standard deviation (as reported in~\figref{fig:attack-comp}).
For each pair $(\id, \w)$, we run the games $\Guessgzero$ and $\Guessgone$ as
specified in~\figref{fig:guess-games}. 
We record the results for guessing budgets of $\q\in\{1,\, 10,\, 10^2,\, 10^3\}$. We repeat
each of the experiments $5$ times and report the averages
in~\figref{fig:attack-comp}.

For \hpbp, we compared implementations using hash prefixes of lengths
$\prefixlen\in\{12, 16, 20\}$. We use the SHA256 hash function with a salt, though
the choice of hash function does not have a noticeable impact on the results. 

For $\ouralgo$, we used interval tree data structures to store the leaked
passwords in $\breachDB$ for fast retrieval of $\Abar(b)$. We used
$\nbuckets=2^{30}$ buckets, and the hash function $f$ is set to
$f(x) = \Ht{30}(x)$, the 30-bit prefix of the (salted) SHA256
hash of the password. 

\paragraph{Attack strategy.}
The attacker's goal is to maximize its success in winning the games
$\Guessgzero$ and $\Guessgone$. In \eqref{eq:adv-zero} and \eqref{eq:adv-one} we
outline the advantage of attackers against $\Guessgzero$ and $\Guessgone$, and
thereby specify the best strategies for attacks. $\Guessgzero$ denotes the
baseline attack success rate in a scenario where the attacker does not have
access to bucket identifiers corresponding to users' passwords. Therefore
the best strategy for the attacker $\advA$ is to output the $\q$ most probable passwords
according to its  knowledge of the password distribution.

The optimal attack strategy for $\advB$ in $\Guessgone$ will be to find a list
of passwords according to the following equation, 
\bnm\argmax_{\substack{\wss{q}\\
    b\in\B((u,\w_i))}} \sum_{i=1}^q
\frac{\CondProb{W=\w_i}{U=u}}{|\B((u,\w_i))|},
\enm 
where the bucket identifier $b$ and user identifier $\user$ are provided to the
attacker. This is equivalent to taking the top-$\q$ passwords in the set $\A(b)$
ordered by $\CondProb{W=\w}{U=\user}/|\B((u,
\w))|$.  

We compute the list of guesses outputted by the attacker for a user $\user$ and bucket $b$
in the following way. For the compromised users, i.e., if
$(\user,\cdot)\in\breachDB$, the attacker first considers the passwords known to be associated to that user and the list of $10^4$ targeted guesses
generated based on the credential tweaking attack introduced
in~\cite{pal2019beyond}. If any of these passwords belong to $\A(b)$ they are
guessed first. This step is skipped for uncompromised users.

For the remaining guesses, we first construct a list of candidates $\guesses$
consisting of all $436$ million passwords present in the breached database
$\breachDB$ sorted by their frequencies, followed by $500\times10^6$ passwords
generated from the $3$-gram password distribution model~$\distn$.
Each password $\w$ in $\guesses$ is assigned a weight $\dists(\w)/|\B((u,\w))|$ (See~\secref{sec:fsb} for details on $\dists$ and $\distn$). The
list $\guesses$ is pruned to only contain unique guesses. Note $\guesses$
is constructed independent of the username or bucket identifier, and it is reordered based
on the weight values. Therefore, it is constructed once for each bucketization
strategy.  Finally, based on the bucket identifier $b$, the remaining guesses are chosen from
$\{\A(b)\cap(u,\w)\mid \w\in\guesses\}$ in descending order of weight.

For the $\hpbp$ implementation, each password is mapped to one bucket, so
$|\B(\w)|=1$ for all $\w$. For $\ouralgo$, $|\B(\cdot)|$ can be calculated using
the equation in~\thref{thm:bucket-entropy}.

Since we are estimating the values to be used in the equation, the attack is no longer optimal. However, the attack we use still performs quite well against existing \ccc protocols, which already shows that they leak too much information. An optimal attack can only perform better. 

\begin{figure*}[th]
  \footnotesize
\begin{tabular}[t]{llrr|rrrr|rrrr}\toprule
	\multirow{2}{*}{Protocol}
	& \multirow{2}{*}{Params}
	&\multicolumn{2}{c|}{Bucket size}
	& \multicolumn{4}{c|}{Uncompromised} 
	& \multicolumn{4}{c}{Compromised}\\
	&&Avg&max& $\q=1$ & $\q=10$ & $\q=10^2$ & $\q=10^3$ & $\q=1$ & $\q=10$ & $\q=10^2$ & $\q=10^3$\\\midrule
	Baseline &N/A&N/A&N/A&0.7 $(\pm 0.1)$& 1.5 $(\pm 0.1)$ & 2.9 $(\pm 0.3)$ & 5.8 $(\pm 0.4)$& 41.1 $(\pm 0.4)$ &  51.1 $(\pm 0.8)$ &  53.3 $(\pm 0.9)$ &  55.7 $(\pm 1.0)$\\ \midrule 
	\multirow{3}{*}{\hpbp} & $\prefixlen=20^\ddagger$& 413 & 491 & 32.9 $(\pm 0.5)$&  49.5 $(\pm 0.3)$&  62.5 $(\pm 0.4)$&  71.1 $(\pm 0.5)$ & 67.3 $(\pm 0.8)$& 74.5 $(\pm 0.6)$ & 79.4 $(\pm 0.6)$& 82.9 $(\pm 0.4)$\\ 
	& $\prefixlen=16^\dagger$ & 6602 & 6891 & 17.9 $(\pm 0.5)$&  33.4 $(\pm 0.6)$&  47.3 $(\pm 0.3)$&  59.7 $(\pm 0.2)$& 61.1 $(\pm 0.9)$& 67.4 $(\pm 0.8)$& 73.6 $(\pm 0.6)$& 78.2 $(\pm 0.7)$\\
	&$\prefixlen=12$ & 105642 & 106668 & 8.2 $(\pm 0.4)$&  17.5 $(\pm 0.6)$&  30.7 $(\pm 0.6)$&  44.4 $(\pm 0.4)$& 56.3 $(\pm 1.0)$& 60.8 $(\pm 1.0)$& 66.5 $(\pm 0.8)$& 72.3 $(\pm 0.6)$\\\midrule
	\multirow{4}{*}{\ouralgo} & $\qbar=1$& 83 & 122 & 0.7 $(\pm 0.1)$&  4.7 $(\pm 0.4)$&  69.8 $(\pm 0.5)$&  71.1 $(\pm 0.5)$& 53.7 $(\pm 0.9)$& 55.7 $(\pm 0.9)$& 82.6 $(\pm 0.4)$& 83.0 $(\pm 0.4)$\\
	&$\qbar=10$& 852 & 965 & 0.7 $(\pm 0.1)$&  1.5 $(\pm 0.1)$&  5.3 $(\pm 0.3)$&  70.8 $(\pm 0.5)$&  52.8 $(\pm 0.9)$& 54.2 $(\pm 1.0)$& 56.0 $(\pm 0.9)$& 83.0 $(\pm 0.4)$	\\
	&$\qbar=10^2$& 6299 & 6602 & 0.7 $(\pm 0.1)$&  1.5 $(\pm 0.1)$&  2.9 $(\pm 0.3)$&  8.0 $(\pm 0.4)$& 51.9 $(\pm 0.8)$& 53.8 $(\pm 0.9)$& 54.8 $(\pm 1.0)$& 57.1 $(\pm 1.0)$\\ 
	&$\qbar=10^3$& 25191 & 25718 & 0.7 $(\pm 1.0)$&  1.5 $(\pm 0.1)$&  2.9 $(\pm 0.3)$&  5.8 $(\pm 0.4)$& 51.4 $(\pm 0.9)$& 53.2 $(\pm 0.9)$& 54.7 $(\pm 1.0)$& 55.9 $(\pm 0.9)$\\ 
	\bottomrule
\end{tabular}\vspace{1em}
  \flushleft{$^\ddagger$ \hibp uses $\prefixlen=20$ for its password-only \ccc service. \;\;
    $^\dagger$ \google uses $\prefixlen=16$ for username-password \ccc service. 
  }
  \caption[comparison]{Comparison of attack success rate given $\q$ queries on
    different password-only \ccc settings. All success rates are in percent (\%)
    of the total number of samples (25,000). The standard deviations across the 5 independent experiments of 5,000 samples each are given in the parentheses. Bucket size, the number of passwords
    associated to a bucket, is measured on a random sample of 10,000
    buckets. 
  }
  \label{fig:attack-comp}
\end{figure*}

\paragraph{Results.} 
We report the success rates of the attack simulations
in~\figref{fig:attack-comp}. The baseline success rate (first row) is the
advantage $\advzero$, computed using the same attack strategy stated above
except 
with no information about the bucket identifier.  The following rows record
the success rate of the attack for \hpbp and \ouralgo with different parameter choices.
The estimated security loss ($\secloss_\q$) can be calculated by subtracting the baseline success rate from the \hpbp and \ouralgo attack success rates.

The security loss from using \hpbp is large, especially for previously uncompromised users.
Accessibility to the $\prefixlen=20$-bit hash prefix, used by
\hibp~\cite{HIBP:v2}, allows an attacker to compromise 32.9\% of previously uncompromised
users in just one guess. In fewer than $10^3$ guesses, that attacker can
compromise more than 70\% of the accounts (12x more than the baseline success
rate with $10^3$ guesses). Google Password Checkup (\gpc) uses $\prefixlen=16$
for its username-password \ccc service. Against $\gpc$, an attacker only needs 10 guesses per account to compromise 33\% of accounts. Reducing the prefix length
$\prefixlen$ can decrease the attacker's advantage. However, that would also increase
the bucket size. As we see for $\prefixlen=12$, the average bucket size is 105,642,
so the bandwidth required to perform the credential check would be
high.

\ouralgo resists guessing attacks much better than $\hpbp$ does. For $\q\le\qbar$ the
attacker gets no additional advantage, even with the estimated password
distribution $\dists$. 
The security loss for \ouralgo when $\q>\qbar$ is much smaller than that of \hpbp, even with smaller bucket sizes. For example, 
the additional advantage over the baseline against \ouralgo with $\q=100$ and
$\qbar=10$ is only 2.4\%, despite \ouralgo also having smaller bucket
sizes than \hpbp with $\prefixlen=16$. Similarly for $\qbar=100$,
$\;\secloss_{10^3} = 2.2\%$. This is because the conditional distribution of
passwords given an \fsbp bucket identifier is nearly uniform, making it
harder for an attacker to guess the correct password in the bucket $\A(b)$ in
$\q$ guesses. 

For previously compromised users --- users present in $\breachDB$ --- even the
baseline success rate is very high: 41\% of account passwords can be guessed in 1
guess and 56\% can be guessed in fewer than 1,000 guesses. The advantage is
supplemented even further with access to the hash prefix.  As per the guessing strategy,
the attacker first guesses the leaked passwords that are both associated to the user and
 in $\A(b)$. This turns out to be very effective.
Due to the high baseline success rate the relative increase is low; nevertheless, in
total, an attacker can guess the passwords of 83\% of previously compromised
users in fewer than 1,000 guesses.  For \fsbp, the security loss for compromised users is comparable
to the loss against uncompromised users for $\q\le\qbar$. Particularly for
$\qbar=10$ and $\q=100$, the attacker's additional success for a previously compromised user is only 2.7\% higher than the baseline. Similarly, for
$\qbar=100$ an attacker gets at most 1.4\% additional advantage for a guessing
budget of $\q$=1,000.  Interestingly, \fsbp performs significantly worse for
compromised users compared to uncompromised users for $\q=1$. This is because
the \fsbp bucketing strategy does not take into account targeted password
distributions, and the first guess in the compromised setting is based on the
credential tweaking attack. 

In our simulation, previously compromised users made up around 46\% of the test
set. We could proportionally combine the success rates against uncompromised and compromised users to obtain an overall attack success rate. However, it is unclear what the actual proportion would be in the real world, so we choose
not to combine results from the two settings. 

\paragraph{Password policy experiment.} 
In the previous set of experiments, we assumed that the \ccc server and the attacker
use the same estimate of the password distribution. To explore a situation in which the attacker
has a better estimate of the password distribution than the \ccc server, we simulated a
website which enforces a password policy. We assume that the policy is known to
the attacker but not to the \ccc server.  

For our sample password policy, we required that passwords have at least 8
characters and that they must not be on Twitter's banned password
list~\cite{twitterbanned}. The test samples are drawn from $\sptestdata$,
username-password pairs from $\testdata$ where passwords follow this policy. The attacker is also given the ability to tailor their guesses to this
policy. The server still stores all passwords in $\breachDB$, without regard to
this policy. Notably, the $\ouralgo$ scheme relies on a good estimate of the
password distribution to be effective in distributing passwords evenly across
buckets. Its estimate, when compared to
the distribution of passwords in $\sptestdata$, should be less accurate than it was in the regular simulation, when
compared to the password distribution from $\testdata$.

We chose the parameters $k=16$ for \hpbp and $\bar{q}=100$ for \ouralgo, because
they were the most representative of how the \hpbp and \ouralgo bucketization
schemes compare to each other. These parameters also lead to similar bucket
sizes, with around 6,500 passwords per bucket. Overall, we see that the success
rate of an attacker decreases in these simulations compared to the general
  experiments (without a password policy). This is because after removing popular
passwords, the remaining set of passwords that we can choose from has higher
entropy, and each password is harder to guess.  \ouralgo still 
defends much better against the attack than \hpbp does, even though the password
distribution estimate used by the \ouralgo implementation is quite inaccurate, 
especially at the head of the distribution. The inaccuracy stems from  \ouralgo assigning larger
probability estimates to passwords that are banned according to the password 
policy. 

We also see that due to the inaccurate estimate by the \ccc server for \ouralgo, we start to see some security loss for an adversary with guessing budget $q=100$. In the general simulation, the password estimate $\dists$ used by the server was closer to $\dist$, so we didn't have any noticeable security loss where $q\leq \bar{q}$.

\begin{figure}[t]
  \centering\footnotesize
\begin{tabular}[t]{l|rrrr|rrrr}\toprule
	\multirow{2}{*}{Protocol}
	& \multicolumn{4}{c|}{Uncompromised} 
	& \multicolumn{4}{c}{Compromised}\\
	& $\q=1$ & $10$ & $10^2$ & $10^3$ & $\q=1$ & $10$ & $10^2$ & $10^3$\\\midrule
	Baseline& 0.1 & 0.5 & 1.3 & 3.4 & 42.2 & 49.0 & 49.8 & 51.1\\ \midrule 
	{\hpbp} ($\prefixlen=16$) & 12.6 & 25.9 & 36.3 & 48.9 & 54.6 & 59.9 & 65.9 & 70.3  \\\midrule
	{\ouralgo} ($\qbar=10^2$) & 0.1 & 0.5 & 1.5 & 13.2 & 49.2 & 50.0 & 50.4 & 54.9 \\\bottomrule
\end{tabular}
  \caption{Attack success rate (in \%) comparison for \hpbp with
    $\prefixlen=16$ (effectively \gpc) and \fsbp with $\qbar=10^2$ for password
    policy simulation. The first row records the baseline success rate
    $\advzero(\q)$. There were 5,000 samples each from the uncompromised and compromised settings.}
  \label{fig:site-policy-sim}
\end{figure}



%% file: performance.tex
\section{Performance Evaluation}\label{sec:perf}
We implement the different approaches to checking compromised
credentials and evaluate their computational overheads. For fair comparison, in
addition to the algorithms we propose, \fsbp and \idbp, we also implement \hibp
and \gpc with our breach dataset. 

\paragraph{Setup.}
We build \ccc services as serverless web applications that provide REST APIs.  We
used AWS Lambda  \cite{lambda} for 
the server-side computation and Amazon DynamoDB \cite{dynamoDb}
to store the data. The benefit of using
AWS Lambda is it can be easily deployed as Lambda@Edge and integrated with
Amazon's content delivery network (CDN), called CloudFront \cite{cloudfront}. (\hibp uses
Cloudflare as CDN to serve more than 600,000 requests per day \cite{cloudflare:blog}.)  We used
Javascript to implement the server and the client side functionalities. The server is implemented as a Node-JS app.  We provisioned the Lambda workers to have a maximum of 
 3~GB of memory. For cryptographic operations, we 
used a Node-JS library called Crypto~\cite{crypto-nodejs}.

For pre-processing and pre-computation of the data we used a desktop with an Intel Core i9
processor and 128~GB RAM.  Though some of the computation (e.g., hash
computations) can be expedited using GPUs, we did not use any for our
experiment. We used the same machine to act as the client. The round trip network latency
of the Lambda API from the client machine is about 130~milliseconds. 

The breach dataset we used is the one described in~\figref{fig:dataset-stats}. It contains 436 million unique passwords and 1,317 million unique username-password pairs. 

To measure the performance of each scheme, we pick 20 random passwords from the
test set $\testdata$ and run the full \ccc protocol with each one. We report the
average time taken 
for each run in~\figref{fig:latency}. In the figure, we also give the breakdown
of the time taken by the server and the client for different operations. The
network latency had very high standard deviation (25\%), though all other
measurements had low ($<1\%$) standard deviations compared to their mean values.

\paragraph{$\hibp$.}
The implementation of \hibp is the simplest among the four schemes. The set of
passwords in $\breachDB$ is hashed using SHA256 and split into $2^{20}$ buckets
based on the first 20~bits of the hash value (we picked SHA256 because we also 
used the same for \fsbp). Because the
bucket sizes in \hibp are so small ($<500$), each bucket is stored as a single value in a
DynamoDB cell, where the key is the hash prefix. For larger leaked datasets, each bucket can be split into
multiple cells.  The client sends the 20~bit prefix of the SHA256 hash of their
password, and the server responds with the corresponding bucket.

Among all the protocols \hibp is the fastest (but also weakest in terms of
security). It takes only 220 ms on average to complete a query over WAN. Most of the time
is spent in round-trip network latency and the query to
DynamoDB. The only cryptographic operation on the client side is a SHA256
hash of the password, which takes less than 1~ms.

\paragraph{$\fsbp$.} The implementation of $\fsbp$ is more complicated than that
of $\hibp$. Because we have more than 1 billion buckets for $\fsbp$
and each password is replicated in potentially many buckets,
storing all the buckets explicitly would require too much storage
overhead. We use interval trees~\cite{wiki:intervaltrees} to quickly recover the passwords in a bucket
without explicitly storing each bucket. 
Each password $\w$ in the breach database is represented as an
interval specified by $\Bfsb(\w)$. We
stored each node of the tree as a separate cell in
DynamoDB. We retrieved the intervals
(passwords) intersecting a particular value (bucket identifier) by querying the nodes stored in DynamoDB. $\fsbp$ also
needs an estimate of the password distribution to get the interval range for a tree. We use $\dists$ as described
in~\secref{sec:hpb}. The description of $\dists$ takes 8.9~MB of space that
needs to be included as part of the client side code. This is only a one-time bandwidth cost during client installation. The client would then need to store the description to use. 

The depth of the interval tree is $\log\dblen$, where $\dblen$ is the number of
intervals (passwords) in the tree. Since each node in the tree is stored as a
separate key-value pair in the database, one client query requires $\log\dblen$ queries to
DynamoDB. To reduce this cost, we split the interval tree into $r$ trees over
different ranges of intervals, such that the $i$-th tree is over the interval
$\left[(i-1)\cdot\floor{{\nbuckets}/{r}},\
  i\cdot\floor{{\nbuckets}/{r}}-1\right]$. The passwords whose bucket
intervals span across multiple ranges are present in all corresponding
trees.  We used $r=128$, as it ensures each tree has around 4 million passwords, and the total storage overhead is less than 1\% more than if we stored one large
tree.

Each interval tree of $4$ million passwords was generated in parallel and took 3 hours in our server. 
 Each interval tree takes 400~MB of storage in DynamoDB, and in total 51~GB of space. \fsbp is the slowest among all the protocols, 
 mainly due to multiple DynamoDB calls, which cumulatively take 273 ms (half of the total
time, including network latency). This can be sped up by using a better
implementation of interval trees on top of DynamoDB, such as storing a whole subtree
in a DynamoDB cell instead of storing each tree node separately. We can also
split the range of the range tree into more granular intervals to reduce each
tree size.  Nevertheless, as the round trip time for \fsbp is small (527 ms), we
leave such optimization for future work. The maximum amount of memory used by the
server is less than 91~MB during an API call.

On the client side, the computational overhead is minimal. The client performs one
SHA256 hash computation.  
The network bandwidth consumed for sending the bucket of hash values from the
server takes on average 558~KB.

\begin{figure}
  \footnotesize
  \begin{tabular}{l|rrr|rr|r|r}\toprule
    \multirow{2}{*}{Protocol} & \multicolumn{3}{c|}{\bf{Client}} & \multicolumn{2}{c|}{\bf{Server}} & Total & Bucket\\
                                   & Crypto & Server call & Comp & DB call & Crypto & time & size\\ 
    \midrule
    HIBP & 1 & 217 & 2 & 40 & --& 220 & 413\\
    FSB & 1 & 524 & 2 & 273 & -- &  527 & 6,602\\
    GPC & 47 & 433 &  9 & 72 & 6 & 489 & 16,121\\
    IDB & 72 & 435 & 10 & 74 & 6 & 517 &  16,122\\
    \bottomrule
  \end{tabular}
  \caption{Time taken in milliseconds to make a C3 API call. 
    The client and server columns contain the time taken to perform client side and server side operations respectively. }
  \label{fig:latency}
\end{figure}

\paragraph{\idbp and \google.} Implementations of $\idbp$ and $\google$ are very
similar.  We used the same platforms --- AWS Lambda and DynamoDB --- to implement
these two schemes. All the hash computations used here are Argon2id with
default parameters, since \google in~\cite{gpc-blog:2019} uses Argon2. During
precomputation, the server computes the Argon2 hash of each username-password
pair and raises it to the power of the server's key $\key$.  These values can be
further (fast) hashed to reduce their representation size, which saves disk
space and bandwidth. However, hashing would make it difficult to rotate server key. 
We therefore store the exponentiated
Argon2 hash values in the database, and hash them further during the online
phase of the protocol.  The hash values are indexed and bucketized based on 
either $\Htl(\user\|\w)$ (for \google) or $\Htl(\user)$ (for \idbp). 
We used $\prefixlen=16$ for both \google and \idbp, as proposed
in~\cite{gpc-blog:2019}. 

We used the secp256k1 elliptic curve.
The server (for both \idbp and \google) only performs one elliptic curve exponentiation,
which on average takes 6~ms. The remaining time incurred is from network latency and calling Amazon
DynamoDB. 

On the client side, one Argon2 hash has to be computed for $\google$ and two for \idbp.  Computing
the Argon2 hash of the username-password pairs takes on an average 20~ms on the
desktop machine. We also tried the same Argon2 hash computation on a personal
laptop (Macbook Pro), and it took 8~ms.  In total, hashing and exponentiation
takes 47~ms for \google, and 72~ms (an additional 25~ms) for \idbp. The
cost of checking the bucket is also higher (compared to \hibp and \ouralgo) due
to larger bucket sizes.

\idbp takes only 28~ms more time on average than \google (due to one extra
Argon2 hashing), while also leaking no additional information about the user's
password. It is the most secure among all the protocols we discussed (should
username-password pairs be available in the leak dataset), and runs in a reasonable time.


%% file: discussions.tex
\section{Deployment discussion}
\label{sec:discussion}
Here we discuss different ways \ccc services can be used and associated
threats that need to be considered. A \ccc service can be queried while creating a password
--- during registration or password change --- to ensure that the new password is not present in a leak. In this setting \ccc is queried from a
web server, and the client IP is potentially not revealed to the server.  This, we
believe, is a safer setting to use than the one we will discuss below. 

In another scenario, a user can directly query a \ccc service. A user can
look for leaked passwords themselves by visiting a web site or using a browser
plugin, such as 1Password~\cite{onepwHIBP} or Password
Checkup~\cite{gpc-blog:2019}.  This is the most prevalent use case of \ccc. For
example, the client can regularly check with a C3 service to proactively safeguard
user accounts from potential credential stuffing attacks.

However, there are several security concerns with this setting. Primarily, the
client's IP is revealed to the \ccc server in this setting, making it easier for
the attacker to deanonymize the user. Moreover, multiple queries from the same
user can lead to a more devastating attack. Below we give two new threat models
that need to be considered for secure deployment of C3 services (where bucket identifiers depend on the password). 

\paragraph{Regular password checks.}
A user or web service might want to regularly check their passwords with \ccc
services.  Therefore, a compromised \ccc server may learn multiple queries from
the same user. For $\fsbp$ the
bucket identifier is chosen randomly, so knowing multiple bucket identifiers for
the same password will help an attacker narrow down the password search space by
taking an intersection of the buckets, which will significantly improve attack success.

We can mitigate this problem for \ouralgo by derandomizing the client side bucket selection using a
client side state (e.g., browser cookie) so the client always selects
the same bucket for the same password. We let 
$c$ be a random number chosen by the client and stored in the browser. To check a password $\w$ with the \ccc server, the client always
picks the $j\thh$ bucket from the range $\B(\w)$, where
$j\gets f(\w\|c)\mod |\B(\w)|$.

This derandomization ensures queries from the same device are deterministic
(after the $c$ is chosen and stored). However, if the attacker can link queries of the
same user from two different devices, the mitigation is 
ineffective. If the cookie is stolen from the client
device, then the security of \fsbp is effectively reduced to that of \hpbp with
similar bucket sizes.

Similarly, if an attacker can track the interaction history between a user and a
\ccc service, it can obtain better insight about the user's passwords. For
example, if a user who regularly checks with a \ccc service stops checking a
particular bucket identifier, that could mean the associated password
is possibly in the most up-to-date leaked dataset, and the attacker can use
that information to guess the user's password(s).

\paragraph{Checking similar passwords.}
Another important issue is querying the \ccc service with multiple correlated
passwords. Some web services, like 1Password, use \hibp to check multiple passwords for a user. 
As shown by prior work, passwords chosen by the same user
are often correlated~\cite{wang2016targeted,das2014tangled,pal2019beyond}. An
attacker who can see  bucket identifiers of multiple correlated passwords can mount a
stronger attack. Such an attack would require
estimating the joint distribution over passwords. We present an initial analysis of this scenario in~\appref{app:correlated}.
  

%% file: relatedwork.tex
\section{Related Work}
\label{sec:relwork}

\paragraph{Private set intersection.}
The protocol task facing C3 services is private set membership, a special case
of private set intersection (PSI)~\cite{meadows1986more,freedman2004efficient}. 
The latter allows two parties to find the intersection 
between their private sets without revealing any additional information. 
Even state-of-the-art PSI protocols do not scale to the sizes needed for our
application. For example, Kiss et al.~\cite{kiss2017private} proposed an efficient PSI protocol
for unequal set sizes based on oblivious pseudo-random functions
(OPRF). It performs well for sets with millions of
elements, but the bandwidth usage scales proportionally to the size of the
leak dataset and so performance is prohibitive in our setting. 
Other efficient solutions to
PSI~\cite{kolesnikov2016efficient,pinkas2015phasing,pinkas2018efficient,chen2017fast}
have similarly prohibitive bandwidth usage. 

Private information retrieval (PIR)~\cite{chor1995private} is another
cryptographic primitive used to retrieve information from a server. Assuming the
server's dataset is public, 
the client can use PIR to privately retrieve the entry corresponding to their
password from the server.  But in our setting we also 
want to protect the privacy of the dataset leak. 
Even if we relaxed that security requirement, the most advanced
PIR schemes~\cite{aguilar2016xpir,olumofin2011revisiting} require exchanging
large amounts of information over the network, so they are not useful for checking leaked passwords.  
PIR with two non-colluding servers can provide better security~\cite{dvir20152} 
than the bucketization-based \ccc schemes,
with communication complexity sub-polynomial in the size of the leaked
dataset. It requires building a \ccc service with two servers
guaranteed to not collude, which may be practical if we assume that the breached credentials are public information. However, with a dataset size of at least 1 billion credentials, the cost of one query is likely still too large to be practical.

\paragraph{Compromised credential checking.} To the best of our knowledge, HIBP
was the first publicly available C3 service. Junade Ali designed the current HIBP
protocol which uses bucketization via prefix hashing to limit leakage. Google's
Password Checkup extends this idea to use PSI, which minimizes the information about
the leak revealed to clients. They also moved to checking username,
password pairs. 

Google's Password Checkup (\google) was described in a paper by Thomas et
al.~\cite{thomas2019protecting}, which became available to us after we began
work on this paper. They introduced the design and implementation of
\google and report on measurements of its initial deployment.
They recognized that their first generation protocol leaks some bits of
information about passwords, but did not analyze the potential impact on
password guessability. They also propose (what we call) the ID-based protocol as
a way to avoid this leakage. Our paper provides further motivation for their
planned transition to it. 

Thomas et al.~point out that password-only C3 services are likely to have high
false positive rates. Our new protocol \fsbp, being in the password-only
setting, inherits this limitation.
That said, should one want to do password-only C3 (e.g., because storing
username, password pairs is considered too high a liability given their utility
for credential tweaking attacks~\cite{pal2019beyond}), \fsbp represents the best known approach.

Other C3 services include, for example, Vericlouds~\cite{leakcheck:vericlouds} and
GhostProject~\cite{leakcheck:ghostproject}. They allow users to register with an email
address, and regularly keep the user aware of any leaked (sensitive)
information associated with that email. Such services send 
information to the email address, and the user implicitly authenticates (proves
ownership of the email) by having access to the email address. These services
are not anonymous and must be used by the primary user.
Moreover, these services cannot be used for password-only C3.

\paragraph{Distribution-sensitive cryptography.} 
Our $\fsbp$ protocol uses an estimate of the distribution of human chosen
passwords, making it an example of distribution-sensitive cryptography, in which
constructions use contextual information about distributions in order to improve
security.  
Previous distribution-sensitive approaches include Woodage et
al.~\cite{woodage2017new}, who introduced a new type of
secure sketch~\cite{dodisetal:2004} for
password typos, and Lacharite et al.'s~\cite{lacharite2018frequency}
frequency-smoothing encryption. While similar in that they use distributional
knowledge, their constructions do not apply in our setting.



%% file: conclusion.tex
\section{Conclusion}
\label{sec:conclusion}
We explore different settings and threat models associated with checking
compromised credentials (C3). The main concern is the secrecy of the user
passwords that are being checked. We show, via simulations, that the existing
industry deployed C3 services (such as \hibp and \gpc) do not provide a satisfying level of 
security.  An attacker who obtains the query to such a C3 service and the
username of the querying user can more easily guess the user's password. We
give more secure C3 protocols for checking leaked passwords and
username-password pairs. We implemented and deployed different C3 protocols on
AWS Lambda and evaluated their computational and bandwidth overhead.  We finish
with several nuanced threat models and deployment discussions that should be
considered when deploying C3 services.


%% file: ack.tex
\section*{Acknowledgments}
We would like to thank the authors of~\cite{thomas2019protecting} for sharing their work with us prior to publication. This work was supported in part by NSF grants CNS-1564102, CNS-1514163, and CNS-1704527.


%% file: appendix.tex
\input{user-pw-correlation}

\section{Bandwidth of \ouralgo}\label{app:fsbp}
To calculate the maximum bandwidth used by \ouralgo, we use the balls-and-bins formula as described in Section~\ref{sec:prelims}. Each password $\w$ is stored in $|\B(\w)|$ buckets, so the total number of balls, or passwords being stored, can be calculated as 
\begin{align*}
m
& = \sum_{w\in\breachDB}|\B(\w)|\\
& = \sum_{w\in \PW_{\qbar}\cap \breachDB} \nbuckets + \sum_{\w\in \breachDB\setminus \PW_{\qbar}} \ceil{\frac{\nbuckets\cdot \dists(\w)}{\qdist}}\\
& \le |\PW_{\qbar}\cap \breachDB|\cdot\nbuckets +
\sum_{\w\in \breachDB\setminus \PW_{\qbar}} \left(\frac{\nbuckets\cdot\dists(\w)}{\qdist} + 1\right)\\
& \le \nbuckets\cdot{\qbar} + \nbuckets\cdot\frac{1}{\qdist} + \dblen 
\end{align*}
The first equality is obtained by replacing the definition of $\B(\w)$; the
second inequality holds because $\ceil{x} \le x + 1$; the third inequality holds because $S\subseteq W$.

The number of bins $n = \nbuckets$, and $m > n\log n$, if $\qbar > \log n$. Therefore, the maximum bucket size for \ouralgo would with high probability be no more than $2\cdot\left({\qbar} + \frac{1}{\qdist} + \frac{\dblen}{\nbuckets}\right)$.


%% file: user-pw-correlation.tex
\section{Correlation between username and passwords}
\label{app:user-pw-correlation}
In~\secref{sec:prelims}, we choose to model the username and password choices of previously uncompromised users independently. 

To check whether this assumption would be valid or not, we randomly sampled $10^5$ username-password pairs from the dataset used in~\secref{sec:empirical} and calculated the Levenshtein edit distance between each username and password in a pair. We have recorded the result of this experiment in~\figref{fig:edit-distance}.

\begin{figure}
  \centering\footnotesize
  \begin{tabular}{r|r}\toprule
    Distance & \% \\
    \midrule
    0 & 1.2 \\ 
    $\leq 1$ & 1.7 \\
    $\leq 2$ & 2.3 \\ 
    $\leq 3$ & 3.1 \\ 
    $\leq 4$ & 4.6 \\\bottomrule
  \end{tabular}
  \caption{Statistics on samples with low edit distance between username and password, as a percentage of a random sample of $10^5$ username-password pairs. }
  \label{fig:edit-distance}
\end{figure}

We found that the mean edit distance between a username and password was 9.4, while the mean password length was 8.4 characters and the mean username length was 10.0 characters. This supports that while there are some pairs where the password is almost identical to the username, a large majority are not related to the username at all. 

The statistics on edit distance between username and password in our dataset are similar to the statistics in the dataset used by Wang et al.~\cite{wang2016targeted}, who determined that approximately 1-2\% of the English-website users used their email prefix as their password.

This data does not prove that usernames and passwords are independent. However, even if an attacker gains additional advantage in the few cases where a user chooses their username as their password, the overwhelming majority of users have passwords that are not closely related to their usernames.


%% file: fsbp-proof.tex
\section{Proof of Theorem~\ref{thm:bucket-entropy}}
\label{sec:bucket-entropy-proof}

	First we calculate the general form of the $\Guessgone_{\Bfsb}$ advantage. Then, we show that for $\q\leq \qbar$, $\advone{\Bfsb}(\q)=\advzero(\q)$, and we bound the difference in the advantages for the games when $\q > \qbar$. 
	
	\begin{align*}
	\advone{\Bfsb}(\q)
	=&\sum_{\user}\sum_{b}
	\max_{\substack{\wss{\q}\\\in\A(b)}} \sum_{i=1}^q \frac{\Prob{W=\w_i\wedge U=\user}}{|\Bfsb(\w_i)|}\\
	=&\sum_{b}	\max_{\substack{\wss{\q}\\\in\A(b)}} \sum_{i=1}^q \frac{\dists(\w_i)}{|\Bfsb(\w_i)|}
	\end{align*}
	The second step follows from the independence of usernames and passwords in the uncompromised setting.
	
	We will use $\PW_{\qbar}$ to refer to the top $\qbar$ passwords according to password distribution $\dists=\distw$, and $\w_{\qbar}$ to refer to the $\qbar$th most popular password according to $\dists$.
	
	For $\w\in \PW_{\qbar}$, we can calculate the fraction in the summation exactly as 
	$\frac{\dists(\w)}{|\Bfsb(\w)|} = \frac{\dists(\w)}{\nbuckets}$. 
	
	For any other 
	$\w\in \PW\setminus \PW_{\qbar}$, we can bound the fraction using the bound on the number of buckets a password is placed in. 
	\begin{equation*}
	\frac{\nbuckets\cdot\dists(w)}{\dists(\w_{\qbar})}\leq |\Bfsb(\w)| < \frac{\nbuckets\cdot\dists(w)}{\dists(\w_{\qbar})}+1.
	\end{equation*}
	We can use the lower bound on $|\Bfsb(\w)|$ to find that
	\begin{equation*}
	\frac{\dists(\w)}{|\Bfsb(\w)|}\leq \frac{\dists(\w_{\qbar})}{\nbuckets}.
	\end{equation*}
	Using the upper bound on $|\Bfsb(\w)|$,
	\begin{equation*}
	\frac{\dists(\w)}{|\Bfsb(\w)|} > \frac{\dists(\w)}{\frac{\nbuckets\cdot\dists(\w)}{\dists(\w_{\qbar})}+1}
	=\frac{\dists(\w)\cdot \dists(\w_{\qbar})}{\nbuckets\cdot \dists(\w)+\dists(\w_{\qbar})}
	=\frac{\dists(\w_{\qbar})}{\nbuckets+\frac{\dists(\w_{\qbar})}{\dists(\w)}}
	\end{equation*}
	Since the values of $\frac{\dists(\w)}{|\Bfsb(\w)|}$ are always larger for $\w\in \PW_{\qbar}$, the values of $\wss{\q}$ chosen for each bucket will be the top $\qbar$ passwords overall, along with the top $\q-\qbar$ of the remaining passwords in the bucket, ordered by $\frac{\dists(\cdot)}{|\Bfsb(\cdot)|}$. 
	
	To find an upper bound on $\advone{\Bfsb}(\q)$, 
	\begin{align*}
	\sum_{b}&	\max_{\substack{\wss{\q}\\\in\A(b)}} \sum_{i=1}^q \frac{\dists(\w_i)}{|\Bfsb(\w_i)|}\\
	& \le \sum_{b}\left(\sum_{w\in \PW_{\qbar}}\frac{\dists(\w)}{\nbuckets} + (\q-\qbar)\frac{\dists(\w_{\qbar})}{\nbuckets}\right)\\
	& = \lambda_{\qbar} + (q-\qbar)\cdot\dist_{\qbar}
	\end{align*}
	For $\q\leq \qbar$, we have $\advone{\Bfsb}(\q) \le \lambda_{\qbar}$.

	To find a lower bound on $\advone{\Bfsb}(\q)$, let $\w_{\qbar+1}^*,\dots,\w_{\q}^*$ be the $\q-\qbar$ passwords in $\A(b)\setminus \PW_{\qbar}$ with the highest probability of occurring, according to $\dists(\cdot)$. 
	\begin{align*}
	\sum_{b}	&\max_{\substack{\wss{\q}\\\in\A(b)}} \sum_{i=1}^q \frac{\dists(\w_i)}{|\Bfsb(\w_i)|}\\
	& > \sum_{b}\left(\sum_{\w\in \PW_{\qbar}}\frac{\dists(\w)}{\nbuckets} + \sum_{i=\qbar+1}^{\q} \frac{\dists(\w_{\qbar})}{\nbuckets+\frac{\dists(\w_{\qbar})}{\dists(\w_i^*)}}  \right)\\
	& \geq \lambda_{\qbar} +\sum_{i=\qbar+1}^{\q}\ceil{\frac{\nbuckets\cdot\dists(\w^*_i)}{\dists(\w_{\qbar})}}\cdot \frac{\dists(\w_{\qbar})}{\nbuckets+\frac{\dists(\w_{\qbar})}{\dists(\w^*_i)}} \\
	& \geq \lambda_{\qbar}
	+\sum_{i=\qbar+1}^{\q}\frac{\nbuckets\cdot\dists(\w^*_i)}{\nbuckets+\frac{\dists(\w_{\qbar})}{\dists(\w^*_i)}}
	\geq \lambda_{\qbar}
	+\sum_{i=\qbar+1}^{\q}\frac{\dists(\w^*_i)}{1+\frac{\dists(\w_{\qbar})}{\dists(\w^*_i)\cdot\nbuckets}}\\
	& \geq \lambda_{\qbar}+\sum_{i={\qbar}+1}^{\q} \dists(\w^*_i)/2 \geq \lambda_{\qbar} +
	(\lambda_{\q}-\lambda_{\qbar})/2 = \frac{\lambda_{\q} + \lambda_{{\qbar}}}{2}
	\end{align*}
	Therefore, $\secloss_{\q} \geq \frac{\lambda_{\q} - \lambda_{\qbar}}{2}$.
	
	Note, for every password to be assigned to a bucket,
	$\nbuckets \ge \dists(\w_{\qbar})/\dists(\w)$, or for all $\w\in \PW$,
	$\frac{\dists(\w_{\qbar})}{\dists(\w)\cdot\nbuckets} \le 1$.  

%% file: correlated.tex
\section{Attacks on Correlated Password Queries}
\label{app:correlated}
\balance
An adversary might gain additional advantage in guessing passwords underlying
\ccc queries when queries are correlated. For example, when creating a new
password, a client might have to generate multiple passwords until the chosen
password is not known to be in a leak.  These human-generated passwords are
often related to each other.  Users also pick similar passwords across different
websites~\cite{das2014tangled,pearman2017let,pal2019beyond,wang2016targeted}.
If such passwords are checked with a \ccc server (maybe by a password
manager~\cite{onepwHIBP}), and the attacker could identify multiple queries from
the same user (for example, by joining based on the IP address of the client),
then the attacker could mount an attack on the correlated queries. As we described, the adversary does need a lot of information to mount such an attack, but the idea is worth exploring, since attacks on correlated queries have not been analyzed before.

Let $\{\transform_{(u, \w)}\}$ be a family of distributions, such that for a given
$u\in\users$, $\w\in\PW$, $\transform_{(u, \w)}$ models a probability distribution across all passwords related to
$\w$ for the user $u$. 
For example, the probability of user $u$ choosing a password $\w_2$ given that they already have password $\w_1$ is $\transform_{(u, \w_1)}(\w_2)$.

The attack game for correlated password queries is given
in \figref{fig:corr-game}.  A client first picks a password $\w_1$ for some web
service and learns that the password is present in a leaked data. The client
then picks another password $\w_2$, potentially correlated to $\w_1$, that is
not known to be in a leak and is accepted by the web service. (For simplicity,
we only consider two attempts to create a password. However, our analysis can easily be
 extended to more than two attempts.)  In the game, the password $\w_2$ is
chosen from the set of passwords not stored by the server, according to the
distribution of passwords from the transformation of $\w_1$. The adversary,
given the buckets $b_1$ and $b_2$, tries to guess the final password, $\w_2$.

\begin{figure}[t]
  \centering
  \fpage{0.25}{
    $\underline{\corrg_{\B}^{\advA}}(\q)$\\
    $(u,\w_1)\genfrom{\dist} \users \times \PW$\\
    $\w_2\genfrom{\tau_{(u, \w_1)}} \PW\setminus\breachDB_w$\\
    $b_1 \gets \B(\w_1);\;\; b_2 \gets \B(\w_2)$\\
    $\{\typo_1,\ldots,\typo_{\q}\} \gets \advA(u, b_1,b_2)$\\
    return $\w_2\in \{\typo_1,\ldots,\typo_{\q}\}$
  }
  \caption{A game to describe a simple correlated password query scenario. Here, we let $\breachDB_w$ be the set of all passwords in the breach dataset.}
  \label{fig:corr-game}
\end{figure}

To find the most likely password given the buckets accessed (the maximum a posteriori estimation), an adversary would want to calculate the following:
\begin{align*}
  \setlength{\abovedisplayskip}{4pt}
  \setlength{\belowdisplayskip}{4pt}
&\argmax_{\w}\; \CondProb{w_2=\w}{b_1,b_2}\\
& \mytab\mytab\mytab=\argmax_{\w}\; \CondProb{b_1,b_2}{w_2=\w}\cdot \frac{\Prob{w_2=\w}}{\Prob{b_1,b_2}}\\
& \mytab\mytab\mytab=\argmax_{\w}\; \CondProb{b_1,b_2}{w_2=\w}\cdot \Prob{w_2=\w}\;.
\end{align*}
Note that we view $b_1, b_2$ as fixed values for the two buckets, not random variables, but we use the notation above to save space. 
We can separate $\CondProb{b_1,b_2}{\w_2=\w}$ into two parts.
\begin{align*}
\CondProb{b_1,b_2}{w_2=\w}
&=\CondProb{b_2}{w_2=\w}\cdot \CondProb{b_1}{w_2=\w,b_2}\\
&=\CondProb{b_2}{w_2=\w}\cdot \CondProb{b_1}{w_2=\w}
\end{align*}
The second step follows from the independence of $b_1$ and $b_2$ given $\w_2$. 

We know that the first term $\CondProb{b_2}{w_2=\w}$ will be 0 if the password $\w$ does not appear in bucket $b_2$. For \fsbp, the buckets that do contain $\w$ have an equally probable chance of being the chosen bucket. For $\hpbp$, only one bucket will have a nonzero probability for each password. 
\bnm
\CondProb{b_2}{w_2=\w}=
\begin{cases}
	\frac{1}{|\B(\w)|} &\text{ if }b_2\in \B(\w)\\
	0 &\text{ otherwise}
\end{cases}.
\enm

Then, to find $\CondProb{b_1}{w_2=\w}$, we need to sum over all passwords that are in $b_1$. We define $\creds_w$ as the set of all possible passwords.
\begin{align*}
\CondProb{b_1}{w_2=\w}
&= \sum_{\w_1\in \creds_w}\CondProb{b_1\land \w_1}{w_2=\w}\\
&= \sum_{w_1\in \A(b_1)} \CondProb{\w_1}{w_2=\w}\\
&=\sum_{w_1\in \A(b_1)} \frac{\CondProb{w_2=\w}{\w_1}\cdot \Prob{\w_1}}{\Prob{w_2=\w}}\;.
\end{align*}

Combining the $\argmax$ expression with the equations above, 
 the adversary therefore needs to calculate the following to find the most likely $\w$:
\begin{align}
\argmax_{\w\in\A(b_2)}\;\frac{1}{|\B(\w)|\cdot \Prob{w_2=\w}}\cdot \sum_{\w_1\in \A(b_1)} \CondProb{w_2=\w}{\w_1}\cdot \Prob{\w_1}\nonumber \\
= \argmax_{\w\in\A(b_2)}\;\frac{1}{|\B(\w)|\cdot \Prob{w_2=\w}}\cdot\sum_{\w_1\in \A(b_1)} \transform_{(u,w_1)}(\w)\cdot\Prob{\w_1}. \label{eq:corr-exp}
\end{align}

In practice, it would be infeasible to compute the above values exactly. For one, the set of all possible passwords is very large, so it would be difficult to iterate over all of the passwords that could be in a bucket. We also don't know what the real distribution $\transform_{(u,w)}$ is for any given $u$ and $w$. For our simulations, we estimate the set of all possible passwords in a bucket using the list constructed by the attack from~\secref{sec:empirical}. To estimate
$\CondProb{w_2=w}{\w_1}$, we use the password similarity measure described
in~\cite{pal2019beyond}, transforming passwords into vectors and calculating the dot product of the vectors. 

\begin{figure}
	\centering\footnotesize
	\begin{tabular}[t]{l|l|rrrr}\toprule
		{Protocol} & Attack &$\q=1$ & $10$ & $10^2$ & $10^3$\\\midrule
		{Baseline}& single-query & 0.2 & 1.0 & 2.9 & 6.4 \\ \midrule 
		\multirow{2}{*}{\hpbp ($\prefixlen=16$)} & single-query & 18.8 & 31.9 & 45.9 & 58.4 \\
		& correlated & 8.8 & 10.3 & 13.0 & 26.0 \\\midrule
		\multirow{2}{*}{\ouralgo ($\qbar=10^2$)} & single-query & 0.2 & 1.0 & 2.9 & 8.4 \\
		& correlated & 2.7 & 3.3 & 4.6 & 11.5 \\\bottomrule
	\end{tabular}
	\caption{Comparison of attack success rate given $q$ queries on our correlated password test set. All success rates are in percent (\%) of the total number of samples (5,000) guessed correctly.}
	\label{fig:corr-table}
\end{figure}

To simulate the correlated-query setting, we used the same dataset as
in~\secref{sec:empirical}. We first trim the test dataset $\testdata$ down to users with passwords both present in the leaked dataset and absent from the leak dataset. We then sample 5,000 of these users and randomly choose the first password from those
present in the leaked dataset and the second password from the ones not in the leaked dataset.  
This sampling most closely simulates the situation where users query a \ccc
server until they find a password that is not present in the leaked data.  We
assume, as before, the adversary knows the username of the querying user.  

For the experiment, we give the attacker access to the leak dataset and the buckets associated with
the passwords $w_1$ and $w_2$. Its goal is to guess the second password, $w_2$. The attacker first narrows down the list constructed in the attack from~\secref{sec:empirical} to only passwords in bucket $b_2$. As a reminder, we refer to this list of passwords as $\Abar(b_2)$. The attacker then computes the similarity between
every pair of passwords in $\Abar(b_2)\times\Abar(b_1)$, which is $\Abar(b_1)$
times the complexity of running a single-query attack (as described
in~\secref{sec:empirical}).  It reorders the list of passwords $\Abar(b_2)$ using an estimate of the value in \eqref{eq:corr-exp}.

The results of this simulation are in~\figref{fig:corr-table}.  We also measured the success rate of the baseline and regular single-query attacks on recovering the same passwords $w_2$. 

It turns out that this correlated attack performs significantly worse than the single-query attack when the passwords are bucketized using \hpbp. For \ouralgo, the correlated attack performs better, but not by a large amount. Although there is an improvement in the correlated attack success for $\fsbp$, the overall success rate of the attack is still worse than both attacks against $\hpbp$. 

The overall low success rate of the correlated attacks is likely due to the error in
estimating the password similarity, $\transform_{(\cdot,\w_1)}(\w)$.  Though the
similarity metric proposed by~\cite{pal2019beyond} is good enough for generating
ordered guesses for a targeted attack, it doesn't quite match the type of correlation among passwords used in the test set. Even though we picked two passwords from the same user for each test point, the passwords were generally not that similar to each other. About 7\% of these password pairs had an edit distance of 1, and only 14\% had edit distances of less than 5. The similarity metric we used to estimate $\transform_{(\cdot,\w_1)}(w)$ heavily favors passwords that are very similar to each other. 

The single-query attack against \hpbp does quite well already, so the correlated attack likely has a lower success rate because it rearranges the passwords in $\Abar(b_2)$ according to their similarity to the passwords in $\Abar(b_1)$. In reality, only a small portion of the passwords in the test set are closely related.
On the other hand, the construction of \ouralgo results in approximately equal probabilities that each password in the bucket was chosen, given knowledge of the bucket. We expect that the success rate for the correlated attack against \ouralgo is higher than that of the single-query attack because the reordering helps the attacker guess correctly in the test cases where the two sampled passwords are similar. 

We believe the error in estimation is amplified in the
attack algorithm, which leads to a degradation in performance. If the attacker knew $\transform$ perfectly and could calculate the exact values in \eqref{eq:corr-exp}, the correlated-query attack would perform better than the single-query attack. However, in reality, even if we know that two queries came from the same user, it is difficult to characterize the exact correlation between the two queries. If the estimate is wrong, then the success of the correlated-query attack will not necessarily be better than that of the single-query attack.
Given that our attack did not show a substantial advantage for attackers, it is still an open question to analyze how damaging attacks on correlated queries can be.